\newcommand{\longversion}[1]{#1}
\newcommand{\shortversion}[1]{}

\longversion{
\documentclass[a4paper,USenglish,11pt]{article}

\usepackage[margin=1.15in]{geometry}
}
\shortversion{
\documentclass[sigplan,screen]{acmart}\settopmatter{}

\acmConference[LICS '20]{Proceedings of the 35th Annual ACM/IEEE Symposium on Logic
  in Computer Science (LICS)}{July 8--11, 2020}{Saarbr\"ucken, Germany}
  \acmBooktitle{Proceedings of the 35th Annual ACM/IEEE Symposium on Logic in
  Computer Science (LICS '20), July 8--11, 2020, Saarbr\"ucken, Germany}
\acmYear{2020}
\acmISBN{978-1-4503-7104-9/20/07}
\acmDOI{10.1145/3373718.?????????}
\startPage{1}

\setcopyright{rightsretained}}
\longversion{
\usepackage[numbers,sort&compress]{natbib}
\bibliographystyle{abbrvnat}}
\shortversion{\bibliographystyle{ACM-Reference-Format}}
\usepackage{booktabs}   %
\usepackage{subcaption} %

\usepackage{enumitem}
\usepackage{mathtools}
\DeclarePairedDelimiter\ceil{\lceil}{\rceil}

\newcommand{\problemFont}[1]{\textsc{#1}}

\newcommand{\footnoteitext}[1]{\stepcounter{footnote}
  \footnotetext[\thefootnote]{#1}}

\newcommand{\trash}[1]{} 

\newlength\problemlength
\newcommand\dproblem[3]{%
\begin{center}
\fbox{%
\begin{minipage}{.93	\linewidth}%
\begin{list}{}{\labelwidth\problemlength \labelsep.7em \rightmargin1.5em
\leftmargin\problemlength \advance\leftmargin by3em%
\parsep0ex \itemsep.2ex plus.1ex}
\item[{\sl Problem:\hfill}] {\problemFont{#1}}
\item[{\sl Input:  \hfill}] #2
\item[{\sl Output: \hfill}] #3
\end{list}
\end{minipage}
}
\end{center}
}
\newcommand\eproblem[3]{%
\begin{center}
\fbox{%
\begin{minipage}{.93	\linewidth}%
\begin{list}{}{\labelwidth\problemlength \labelsep.7em \rightmargin1.5em
\leftmargin\problemlength \advance\leftmargin by3em%
\parsep0ex \itemsep.2ex plus.1ex}
\item[{\sl Problem:\hfill}] {\problemFont{#1}}
\item[{\sl Input:  \hfill}] #2
\item[{\sl Question: \hfill}] #3
\end{list}
\end{minipage}
}
\end{center}
}

\usepackage{microtype}
\usepackage{multirow}
\usepackage{xcolor}

\usepackage{hyperref}

\newcommand{\FIX}[1]{#1} %
\newcommand{\FIXCAM}[1]{#1} %
\newcommand{\BUGFIX}[2]{\FIXCAM{#1}}

\usepackage[normalem]{ulem}

\usepackage{amsmath,amssymb,amsthm,amsopn}
\usepackage{relsize}
\usepackage{xspace}

\def\hy{\hbox{-}\nobreak\hskip0pt}

\newtheorem{LEM}{Lemma}
\newtheorem{EX}[LEM]{Example}
\newtheorem{EXa}[LEM]{Example}

\newtheorem{PROP}[LEM]{Proposition}

\newtheorem{REM}[LEM]{Remark}
\newtheorem{OBS}[LEM]{Observation}
\newtheorem{THM}[LEM]{Theorem}
\newtheorem{DEF}[LEM]{Definition}
\newtheorem{COR}[LEM]{Corollary}
\newtheorem{CLAIM}[LEM]{Claim}

\renewenvironment{EX}{\begin{EXa}}{\hfill\ensuremath{\blacksquare}\end{EXa}}
\newenvironment{restatetheorem}[1][\unskip]{%
  \begingroup
  \def\theLEM{\ref{#1}} 
}%
{%
  \addtocounter{LEM}{-1}
  \endgroup
}%

\newcommand{\SEM}{\ensuremath{\mathcal{S}}\xspace}
\newcommand{\ALL}{\ensuremath{\mathrm{ALL}}\xspace}

\newcommand{\Cred}{\problemFont{Cred}}
\newcommand{\Skep}{\problemFont{Skep}}

\DeclareMathOperator{\adef}{def}
\DeclareMathOperator{\stage}{stage}
\DeclareMathOperator{\semistable}{semi-st}
\DeclareMathOperator{\preferred}{pref}
\DeclareMathOperator{\admissible}{adm}

\newcommand{\complexityClassFont}[1]{\ensuremath{\mathrm{#1}}}
\newcommand{\Ptime}{\complexityClassFont{P}}
\newcommand{\PSPACE}{\ensuremath{\textsc{PSpace}}}
\newcommand{\SAT}{\textsc{Sat}\xspace}
\newcommand{\PMC}{\textsc{PMC}\xspace}
\newcommand{\sharpSAT}{\textsc{\#Sat}\xspace}

\newcommand{\QBFSAT}{\textsc{QSat}\xspace}
\newcommand{\PQBF}{\textsc{PQSat}\xspace}
\newcommand{\PASP}{\textsc{PASP}\xspace}
\newcommand{\cASP}{\textsc{ASPCons}\xspace}
\newcommand{\sharpASP}{\textsc{\#ASP}\xspace}
\newcommand{\PAP}{\textsc{PAP}\xspace}
\newcommand{\Circ}{\textsc{Circumscription}\xspace}
\newcommand{\PPAP}{\text{\#}\textsc{PPAP}\xspace}
\newcommand{\PCC}{\protect\ensuremath{\problemFont{\#PCred}}\xspace}

\newcommand{\NP}{\ensuremath{\textsc{NP}}\xspace}

\newcommand{\Q}{\ensuremath{Q}}
\newcommand{\SB}{\{}%
\newcommand{\SM}{\mid}%
\newcommand{\SE}{\}}%

\newcommand{\citex}[1]{\citet{#1}}

\newcommand{\smallalign}[1]{#1
  \setlength{\abovedisplayskip}{3pt}
  \setlength{\belowdisplayskip}{\abovedisplayskip}
  \setlength{\abovedisplayshortskip}{0pt}
  \setlength{\belowdisplayshortskip}{3pt}}

\DeclareMathOperator{\dom}{\mathsf{dom}}%
\DeclareMathOperator{\cntc}{\#\smash{\cdot}}%
\newcommand{\mtext}[1]{\text{\normalfont #1}}
\newcommand{\NAT}{\ensuremath{\mathbb{N}}}
\newcommand{\Nat}{\mathbb{N}} %

\newcommand{\TTT}{\ensuremath{\mathcal{T}}}%
\newcommand{\WWW}{\ensuremath{\mathcal{W}}}%

\DeclareMathOperator{\ep}{EP}
\newcommand{\prog}{\ensuremath{\Pi}}
\newcommand{\hsep}{\leftarrow\,}
\newcommand{\por}{\vee}

\newcommand{\tw}[1]{\mathsf{tw}(#1)}
\newcommand{\pw}[1]{\mathsf{pw}(#1)}
\newcommand{\lint}[2]{\mathit{VarI}(#2)}
\newcommand{\master}[1]{\mathsf{rep}(#1)}
\newcommand{\nint}[2]{\mathit{NodeI}(#2)}
\newcommand{\lbv}[1]{\mathit{VarB}}
\newcommand{\lbvd}[1]{\mathit{VarB3}}
\newcommand{\lbvs}[1]{\mathit{VarBs}}
\newcommand{\lbvv}[1]{\mathit{VarBV}}
\newcommand{\lbvvd}[1]{\mathit{VarBV3}}
\newcommand{\lbvvs}[1]{\mathit{VarBVs}}

\newcommand{\lsat}{\mathit{VarSat}}

\DeclareMathOperator{\tower}{\ensuremath{\mathsf{tow}}}
\DeclareMathOperator{\poly}{\ensuremath{{poly}}}

\newcommand{\ta}[1]{\ensuremath{2^{#1}}}
\newcommand{\Card}[1]{\left|#1\right|}
\newcommand{\CCard}[1]{\|#1\|}

\newcommand{\bterm}[2]{\ensuremath{\mathsf{bit\hy term}({#1},{#2})}}
\newcommand{\local}[1]{\ensuremath{\mathsf{local}({#1})}}
\newcommand{\lit}[2]{\ensuremath{\mathsf{tlit}(#1, #2)}}
\newcommand{\bvv}[3]{\ensuremath{\mathsf{bv}(#1, #2, #3)}}
\newcommand{\bvs}[3]{\ensuremath{\mathsf{bvc}(#1, #2, #3)}}
\newcommand{\ato}[2]{\ensuremath{\mathsf{tvar}(#1, #2)}}
\newcommand{\lnode}[1]{\ensuremath{\delta^{-1}(#1)}}
\newcommand{\bval}[2]{\ensuremath{[\![#1]\!]_{#2}}}

\newcommand{\bvali}[3]{\ensuremath{[\![#1]\!]_{#2,#3}}}

\DeclareMathOperator{\var}{\mathsf{var}}
\DeclareMathOperator{\fvar}{\mathsf{fvar}}
\DeclareMathOperator{\matr}{\mathsf{matrix}}
\newcommand{\bigO}{\ensuremath{{\mathcal O}}}
\newcommand{\eqdef}{\ensuremath{\,\mathrel{\mathop:}=}}

\DeclareMathOperator{\width}{\mathsf{width}}
\DeclareMathOperator{\children}{\mathsf{cld}}

\makeatletter
\newcommand{\pushright}[1]{\ifmeasuring@#1\else\omit\hfill$\displaystyle#1$\fi\ignorespaces}
\newcommand{\pushleft}[1]{\ifmeasuring@#1\else\omit$\displaystyle#1$\hfill\fi\ignorespaces}
\makeatother

\longversion{

\usepackage{authblk}
\title{Lower Bounds for QBFs of Bounded Treewidth}}

\begin{document}

\longversion{
\author[1]{Johannes K. Fichte\thanks{johannes.fichte@tu-dresden.de}}
\author[2,3]{Markus Hecher\thanks{hecher@dbai.tuwien.ac.at}}
\author[3]{Andreas Pfandler\thanks{pfandler@dbai.tuwien.ac.at}}

\affil[1]{International Center for Computational Logic, TU Dresden, Germany}
\affil[2]{Institute of  Logic and Computation, TU Wien,  Austria}
\affil[3]{Institute of Computer Science, University of Potsdam, Germany}
}

\shortversion{\title[Lower Bounds for QBFs of Bounded Treewidth]{Lower Bounds for QBFs of Bounded Treewidth}         %
\author{Johannes K. Fichte}
\affiliation{
  \institution{TU Dresden}            %
  \country{Germany}                    %
}
\email{johannes.fichte@tu-dresden.de}          %

\author{Markus Hecher}
\affiliation{
  \institution{TU Wien}           %
  \country{Austria}                   %
}
\email{hecher@dbai.tuwien.ac.at}         %
\author{Andreas Pfandler}
\affiliation{
  \institution{TU Wien}           %
  \country{Austria}                   %
}
\email{pfandler@dbai.tuwien.ac.at}         %
}

\longversion{\maketitle}

\begin{abstract}
The problem of deciding the validity (\QBFSAT) of quantified Boolean formulas (QBF) is a vivid research area in both theory and practice. 
In the field of parameterized algorithmics, the well-studied graph measure treewidth turned out to be a successful parameter. 
\FIXCAM{A well-known result by Chen~\cite{Chen04a} is that \QBFSAT when parameterized by the treewidth of 
the primal graph  and the quantifier rank of the input formula is fixed-parameter tractable.}
More precisely, the runtime of such an algorithm is polynomial in the formula size and exponential in the treewidth, 
where the exponential function in the treewidth is a tower, whose height is the quantifier rank. 
A natural question is whether one can significantly improve these results and decrease the tower while assuming the 
Exponential Time Hypothesis (ETH). In the last years, there has been a growing interest in the quest of establishing lower bounds under ETH, 
showing mostly problem-specific lower bounds up to the third level of the polynomial hierarchy. Still, an important question is to settle this 
as general as possible and to cover the whole polynomial hierarchy. 
In this work, we show lower bounds based on the ETH for arbitrary QBFs parameterized by treewidth and quantifier rank. 
More formally, we establish lower bounds for \QBFSAT and treewidth, namely, that under ETH there cannot be an algorithm that solves 
\QBFSAT of quantifier rank i in runtime significantly better than $i$-fold exponential in the treewidth and polynomial in the input size. 
In doing so, we provide a reduction technique to compress treewidth that encodes dynamic programming on arbitrary tree decompositions. 
Further, we describe a general methodology for a more fine-grained analysis of problems parameterized by treewidth that are at higher levels of the polynomial hierarchy. 
Finally, we illustrate the usefulness of our results by discussing various applications of our results to problems that are located higher on the polynomial hierarchy, in particular,
various problems from the literature such as projected model counting problems.  

\end{abstract}

\if\shortversion{1}1{
\ccsdesc[500]{Theory of computation~Parameterized complexity and exact algorithms}
  \ccsdesc[500]{Theory of computation~Complexity theory and logic}
  \ccsdesc[300]{Theory of computation~Complexity classes}
\ccsdesc[500]{Mathematics of computing~Graph algorithms}
\ccsdesc[300]{Mathematics of computing~Trees}
\ccsdesc[300]{Computing methodologies~Knowledge representation and reasoning}
\keywords{Parameterized Complexity, ETH, Lower Bounds, Quantified Boolean Formulas, Tree\-width, Pathwidth}  %
}
\fi

\longversion{
\newcommand{\grantsponsor}[3]{#2}
\newcommand{\grantnum}[2]{#2}

}

\shortversion{\maketitle}

\noindent\textbf{Acknowledgments.} {\small
\longversion{This is an extended version of a paper that appeared at the 35th Annual ACM/IEEE Symposium on Logic in Computer Science (LICS 2020).}
This work has been supported by the \grantsponsor{GS100000001}{Austrian Science Fund (FWF)}{fwf.ac.at},
   Grants \grantnum{GS100000001}{P26696}, \grantnum{GS100000001}{P30930-N35}, \grantnum{GS100000001}{P32830}, and \grantnum{GS100000001}{Y698}, and the \grantsponsor{GS100000002}{German Science Fund (DFG)}{dfg.de}, Grant \grantnum{GS100000002}{HO 1294/11-1}. Hecher is also affiliated with the University of Potsdam, Germany. We would like to thank the reviewers as well as Stefan Szeider and Stefan Woltran for their support.
Special appreciation goes to Michael Morak for early discussions.}

\section{Introduction}

\emph{Treewidth}, which was introduced specifically for graph problems
by Robertson and Seymour in a series of
papers~\cite{RobertsonSeymour83,RobertsonSeymour84,RobertsonSeymour85,RobertsonSeymour86,RobertsonSeymour91},
is a popular parameter in the community of parameterized
complexity~\cite{CyganEtAl15,DowneyFellows13,FlumGrohe06} and
\FIXCAM{according to Google Scholar
 mentioned in 20,000
  results (queried on April 27, 2020).
Treewidth is a} combinatorial invariant that renders a large variety of {\sc
  NP}-complete or \#P-complete graph problems
tractable~\cite{BodlaenderKoster08,ChimaniMutzelZey12}. Among these
problems
are for example deciding whether a graph has a Hamiltonian cycle,
whether a graph is 3-colorable, or determining the number of perfect
matchings of a graph~\cite{CourcelleMakowskyRotics01}.
Still, treewidth has also been widely employed for important
\FIX{applications} that are defined on more general input structures
such as \FIX{Boolean} satisfiability (\SAT)~\cite{SamerSzeider10b} and
constraint satisfaction (CSP)~\cite{Dechter06,Freuder85}.
Even problems that are located ``beyond NP'' such as probabilistic
inference~\cite{OrdyniakSzeider13}, problems in knowledge
representation and
reasoning~\cite{GottlobPichlerWei10,PichlerRuemmeleWoltran10,DvorakPichlerWoltran12a}
as well as deciding the validity (\QBFSAT) of quantified Boolean formulas
(QBF)  can be turned tractable using
treewidth\FIX{; for~\QBFSAT we also parameterize by the number of alternating quantifier blocks (quantifier rank)~\cite{Chen04a}}.
\FIX{However, \QBFSAT remains intractable when parameterized by treewidth alone~\cite{AtseriasOliva14a}, which is established using a particular fragment 
of path decompositions for QBF.}
\QBFSAT is also known as the prototypical problem for the polynomial
hierarchy in descriptive complexity~\cite{Grohe12,Immerman99}. Indeed,
an encoding in QBF allows the characterization of problems on certain
levels of the hierarchy using results by Fagin~\cite{Fagin74}.~This
has, for instance, been done for %
reasoning
problems~\cite{EiterGottlob95,EiterGottlob95b,EglyEtAl00}.

\emph{The} meta results on treewidth are the well-known
\emph{Courcelle's theorem}~\cite{Courcelle90} and its \FIXCAM{logspace} %
version~\cite{ElberfeldJakobyTantau10a}, which states that whenever one
can encode a problem into a formula in monadic second order logic
(MSO), then the problem can be decided in time linear in the input
size \FIX{and some function in the treewidth}. While Courcelle's theorem
provides a full framework for classifying problems concerning the existence of a tractable algorithm, its practical application is limited due to potentially huge constants, and
the exponential runtime in the treewidth (upper bound) may result in a
tower of exponents that is far from optimal.
In contrast, the available upper bounds are more immediate for
\QBFSAT:  Chen~\cite{Chen04a} showed that one can decide validity for
a given QBF in time exponential in the treewidth where the treewidth
is on top of a \emph{tower}\footnote{%
  \label{footnote:tower}
  Function~$\tower(\ell,k)$ is a tower of iterated exponentials of
  $2$ of height~$\ell$ with $k$ on top. More precisely, for integer~$k$, we %
  define %
  $\tower: \NAT \times \NAT \rightarrow \NAT$ by $\tower(1,k) = 2^k$
  and $\tower(\ell + 1,k) = 2^{\tower(\ell,k)}$ for all
  $\ell \in \NAT$.
}%
 of iterated exponentials of height that
equals the quantifier rank in the formula. %
Since the quantifier rank required to
encode a problem directly matches the level on which the
problem is located in the polynomial hierarchy, reductions to \FIXCAM{\QBFSAT
seem natural.}
Lampis, Mitsou, and Mengel~\cite{LampisMitsouMengel18} employed this
fact and proposed reductions from a collection of reasoning problems
in AI to \QBFSAT that \FIX{yield} quite precise (up to \FIX{a constant} factor)
upper bounds on the runtime.  In consequence, these results highlight QBF
encodings as a very handy and precise alternative to Courcelle's
theorem.
A natural question is whether one can significantly improve existing
algorithms or establish limits that, unless very bad things happen in
computational complexity theory, an algorithm with a certain runtime
cannot exist.
Lampis, Mitsou, and Mengel also consider this question using
results~\cite{LampisMitsou17} for QBF of quantifier rank two
(2-\QBFSAT).
While these results for the second level are applicable to numerous
important problems, there is also a plethora of interesting problems
that are even harder, which underlines the need for further research
in this direction.

In this paper, we address lower bounds for the runtime of an algorithm
that exploits treewidth in a more general setting. We establish
results for QBFs of bounded treewidth and of \emph{arbitrary} quantifier
rank, thereby providing a novel method to generalize the result for
2-\QBFSAT in a non-incremental \FIXCAM{way}.

\FIXCAM{A way to establish tight lower bounds in parameterized
complexity} theory is to assume the \emph{exponential time hypothesis
  (ETH)~\cite{ImpagliazzoPaturiZane01}} and construct reductions.  ETH is a widely accepted standard
hypothesis in the fields of exact and parameterized algorithms. ETH
states that there is some real~$s > 0$ such that we cannot decide
satisfiability of a given 3\hy \FIXCAM{CNF formula~$F$} in
time~$2^{s\cdot\Card{F}}\cdot\CCard{F}^{\mathcal{O}(1)}$~\cite[Ch.14]{CyganEtAl15}\FIXCAM{, where $\Card{F}$ refers to the number of variables and  $\CCard{F}$ to  
the \emph{size} of~$F$, which is number of variables plus number of clauses in~$F$.}
Recently, Lampis and Mitsou~\cite{LampisMitsou17} established that
2-\QBFSAT ($\exists\forall$\hy\SAT and $\forall\exists$\hy\SAT) cannot
be solved by an algorithm that runs in time single exponential in the
treewidth of the primal graph (\emph{primal treewidth}) when assuming
ETH.
\FIX{The primal graph of a QBF~$\Q$ has as vertices the variables of~$\Q$
and there is an edge between two variables if they occur together in a
clause \FIXCAM{or term}, respectively.}
\FIX{\citex{PanVardi06} \FIXCAM{mention in an earlier work  %
that this extends to 3-\QBFSAT
($\forall\exists\forall$\hy\SAT and~$\exists\forall\exists$\hy\SAT),
and~$\ell$-\QBFSAT, if~$\ell$ is an \emph{odd} number.
But it  does not extend constructively to the case,} where~$\ell$ is even.
Therefore, a new approach is needed to show the complete picture for \QBFSAT. }%
While Marx and Mitsou~\cite{MarxMitsou16} considered certain graph
problems that are located on the third level of the polynomial
hierarchy~\cite{StockmeyerMeyer73}, they emphasize that the classical
complexity results do not provide sufficient explanation why double-
or triple-exponential dependence on treewidth is needed and one
requires quite involved proofs for each problem separately. However, they
state that intuitively the quantifier rank of the problem
definitions are the common underlying reason for being on higher
levels of the polynomial hierarchy and for requiring \FIX{high
dependence on treewidth}.
A natural generalization of the statement to arbitrary QBFs is
formally stated in the following hypothesis.

\begin{CLAIM}%
\label{hyp:teeth}
  \FIXCAM{Under ETH, \QBFSAT} for a closed formula~$\Q$ in prenex normal form
  with $n$ variables, primal treewidth~$k$, and quantifier
  rank~$\ell$ \emph{cannot} be decided in
  time~$\tower(\ell,o(k))\cdot \BUGFIX{\poly(n)}{2^{o(n)}}$.
\end{CLAIM}

\noindent\textbf{Contributions.} %
\noindent In this paper, we prove Claim \ref{hyp:teeth}, which strengthens the importance of
QBF encodings for problems parameterized by
treewidth, %
and establish a general methodology to obtain treewidth lower bounds for
problems \FIX{of} the polynomial hierarchy. %
Our \emph{contributions} are as follows:

\begin{enumerate}[leftmargin=0pt, itemindent=10pt, itemsep=5pt]
\item We consider ETH and \QBFSAT and establish \FIX{the full picture of} runtime lower bounds for
  algorithms parameterized by treewidth in connection to the quantifier rank of the formula.

  We present a reduction \FIX{that significantly
  \emph{compresses}} treewidth and applies to any instance of \QBFSAT %
  without restricting the quantifier rank while only
  assuming ETH. \FIX{Note that this ``compression'' is constructive and independent of the original instance size, which is different from existing methods,~e.g., \cite{LampisMitsou17,MarxMitsou16,PanVardi06}. In fact, compression only depends on the original parameter treewidth.}

\item We provide a novel methodology for a more fine-grained analysis of
  algorithms parameterized by treewidth. \FIXCAM{This methodology relies only on %
the ETH and allows:} %
  \begin{enumerate}[leftmargin=4pt,itemindent=10pt,topsep=2pt,itemsep=2pt]
  \item \FIXCAM{for simply using reductions from QBF to exclude runtime
    results (height of the tower) for treewidth and}
  \item \FIX{for directly concluding lower bounds for projected model counting problems (\PQBF), that is, $\text{\#}\Sigma_\ell\SAT$ and $\text{\#}\Pi_\ell\SAT$~\cite{DurandHermannKolaitis05},
  which serve as canonical problems in
    counting complexity, as well as for various other problems}.

\FIXCAM{
Instead of establishing problem-specific reductions from SAT for problems higher in the polynomial hierarchy,~e.g., \cite{LampisMitsou17,MarxMitsou16}, our reduction and methodology are very general. 
We illustrate its applicability in the showcases in Section~\ref{sec:metho}.\vspace{-.5em}
}

  \end{enumerate}

\end{enumerate}

\noindent\textbf{Novel Techniques.} %
\FIXCAM{We} \FIX{constructively} encode core ideas of a
dynamic programming algorithm on tree decompositions into a QBF that
expresses solving an instance of \QBFSAT by means of a \QBFSAT oracle
of \FIX{\emph{one level}} higher in the hierarchy (self-reduction)
while achieving a certain compression of treewidth. More precisely, we
provide a reduction that reduces any instance~$\Q$ of \QBFSAT of
treewidth~$t$ and quantifier rank~$\ell$ into an instance~$\Q'$ of
\QBFSAT of treewidth~$\bigO(\log t)$ and quantifier rank~$\ell +1$,
while the size of $\Q'$ is linearly bounded in the size
of~$\Q$. \FIX{Notice that the treewidth of~$\Q'$ only depends on the
  treewidth of~$\Q$, but is independent of, e.g., the number of
  variables and quantifier rank of~$\Q$.  \FIX{Hence,} %
\emph{treewidth} of~$\Q'$ is \emph{compressed}
  compared to the original treewidth of~$\Q$.
  \citex{AtseriasOliva14a} cover a related setting: compressing
  pathwidth\FIX{\footnote{\label{foot:pathwidth}\FIX{Pathwidth is
        similar to treewidth, but admits only \emph{certain} tree
        decompositions, whose tree is just a path.}}}  for a fragment
  of path decompositions of QBFs thereby increasing the quantifier
  rank by two.  However, we require a general, constructive method to
  compress the width of \emph{arbitrary} tree decompositions of any
  QBF, thereby increasing quantifier rank by only one, 
and improve their result~(Corollary~\ref{cor:twhardness}).
}

  Our reduction is novel in the following sense:
\begin{enumerate}[leftmargin=0pt, itemindent=10pt, itemsep=5pt]
  \item We use \FIX{a given tree decomposition} to guide the evaluation of
    the considered formula, which allows us to decouple the variables
    sufficiently to decrease treewidth and we thereby achieve exponential \emph{compression} of the parameter treewidth. \FIX{By construction of the reduction, the lower bound results carry over to the \emph{larger} parameter \emph{pathwidth}$^\text{\ref{foot:pathwidth}}$
    and even \emph{$2$-local pathwidth}, where each %
    variable occurs at most twice.
    Note that this direction is by construction and does not hold in general. 
    However, particular novelty lies in encoding essentials of dynamic programming, which will be presented in the more general context of treewidth (tree decompositions).}
  \item In the proof, we use a reduction approach that \FIXCAM{balances} redundancy and structural dependency (captured by \FIX{treewidth or pathwidth}), which
    allows us to apply this method to QBFs of \FIX{\emph{arbitrary}} quantifier
    rank\FIXCAM{, thereby increasing quantifier rank by only one}.
   \item \FIX{
Our approach might help to improve solvers utilizing treewidth, as instances of huge treewidth might become solvable in practice (cf.,~\cite{CharwatWoltran19,FichteEtAl20}) after applying our reduction.
   Indeed, our reduction encodes  dynamic programming on tree decompositions into a Boolean formula, namely, \emph{guessing} of finite states for table entries of decomposition nodes,
   \emph{checking} whether certain entries sustain, and \emph{propagating} entries among different nodes.
   {As this technique, although presented for QBFs, does not explicitly encode quantifier dependencies into the Boolean formula, %
   the technique is \FIXCAM{hopefully of general use.}}
}%
\end{enumerate}

\paragraph{Connection to kernels.} Note that our approach is orthogonal to kernelization as
    kernelizations tackle bounds of the \emph{instance size} by the
    considered parameter, whereas here we target reducing  the parameter itself and not the size of the input instance.

\section{Preliminaries}\label{sec:prelim}
\noindent\textbf{Basics.} For a set~$X$, let $\ta{X}$ be the \emph{power set of~$X$}.
The function~$\tower(\ell,k)$ is defined as in
Footnote~\ref{footnote:tower}.
The domain~$\mathcal{D}$ of a
function~$f:\mathcal{D} \rightarrow \mathcal{A}$ is given
by~$\dom(f)$. By $f^{-1}: \mathcal{A} \rightarrow \mathcal{D}$ we
denote the inverse
function~$f^{-1}:=\SB f(d) \mapsto d \SM d \in \dom(f) \SE$ of a given
function~$f$, if it exists.  To permit operations such as~$f \cup g$
for functions~$f$ and $g$, the functions may be viewed as relations.
\FIXCAM{%
  We use the symbol~``$\cdot$'' as placeholder for a value of an
  argument, which is clear from the context and the actual value is
  negligible.
}
We let~$\NAT$ contain all positive integers and $\NAT_0$ all
non-negative integers. \FIX{Throughout this paper, we refer by~$\log(\cdot)$ to the binary logarithm.}

\smallskip
\noindent \textbf{Computational Complexity.}
\FIX{We assume familiarity with standard notions in computational
computational complexity~\cite{Papadimitriou94},
counting complexity classes~\cite{DurandHermannKolaitis05}, and
parameterized complexity~\cite{CyganEtAl15,DowneyFellows13,FlumGrohe06}.}
We recall some basic notions.
Let $\Sigma$ and $\Sigma'$ be some finite alphabets.  We call
$I \in \Sigma^*$ an \emph{instance} and $\CCard{I}$ denotes the size
of~$I$.  %
Let $L \subseteq \Sigma^* \times \Nat$ and
$L' \subseteq {\Sigma'}^*\times \Nat$ be two parameterized problems. An
\FIX{fpt-reduction~$r$ using~$g$} from $L$ to $L'$ is a many-to-one reduction
from $\Sigma^*\times \Nat$ to ${\Sigma'}^*\times \Nat$ such that for all
$I \in \Sigma^*$ we have $(I,k) \in L$ if and only if
$r(I,k)=(I',k')\in L'$ \FIXCAM{with} $k' \leq g(k)$, \FIXCAM{where~$f,g: \Nat \rightarrow \Nat$ 
are fixed computable functions
such that $r$ is computable in time
$f(k)\cdot\poly(\CCard{I})$}. 
\FIX{We call~$r$ also an \emph{$f$-bounded} fpt-reduction using~$g$ for given~$f$ and~$g$.}

\smallskip
\noindent\textbf{Quantified Boolean Formulas (QBFs).}
We define \emph{Boolean formulas} and their evaluation in the usual
way and \emph{literals} are variables or their negations.  For a
Boolean formula~$F$, we denote by~$\var(F)$ the set of variables
of~$F$. \FIXCAM{Logical operators~$\wedge, \vee, \neg, \rightarrow, \leftrightarrow$
are used in the usual meaning.}
{%
  A \emph{term} is a conjunction of literals and a \emph{clause} is a
  disjunction of literals. $F$ is in \emph{conjunctive normal form
    (CNF)} if $F$ is a conjunction of clauses and $F$ is in
  \emph{disjunctive normal form (DNF)} if $F$ is a disjunction of
  terms.
  In both cases, we identify $F$ by its \FIXCAM{set of clauses or terms, respectively.}
  From now on assume that a Boolean formula is either in CNF or DNF.
  A formula is in \emph{$c$-CNF} or \emph{$c$-DNF} if each set in~$F$
  consists of at most $c$ many literals.
}%
  Let $\ell\geq 0$ be integer. A \emph{quantified Boolean
    formula}~$\Q$ \emph{(in prenex normal form)} is of the form
  $Q_{1} V_1.  Q_2 V_2.\cdots Q_\ell V_\ell. F$ where
  $Q_i \in \{\forall, \exists\}$ for $1 \leq i \leq \ell$ and
  $Q_j \neq Q_{j+1}$ for $1 \leq j \leq \ell-1$; and where $V_i$ are
  disjoint, non-empty sets of Boolean variables with
  $\bigcup^\ell_{i=1}V_i \subseteq \var(F)$; and $F$ is a Boolean
  formula.
  We call $\ell$ the \emph{quantifier rank} of~$Q$ and let
  $\matr(\Q)\eqdef F$.
  Further, we denote the set~$\fvar(Q)$ of \emph{free variables}
  of~$Q$ by
  $\fvar(\Q)\eqdef \var(\matr(\Q)) \setminus (\bigcup^\ell_{i=1}V_i)$.
If~$\fvar(\Q)=\emptyset$, then~$\Q$ is referred to as~\emph{closed},
otherwise we say~$\Q$ is~\emph{open}.  Unless stated otherwise, we
assume open QBFs. 
  The truth (evaluation) of QBFs is defined in the standard way.
  An \emph{assignment} is a mapping~$\iota: X \rightarrow \{0,1\}$
  defined for a set~$X$ of variables.
\FIXCAM{An assignment~$\iota'$ \emph{extends}~$\iota$
  (by~$\dom(\iota')\setminus\dom(\iota)$) if
  $\dom(\iota') \supseteq \dom(\iota)$ and~$\iota'(y) = \iota(y)$ for
  any~$y\in\dom(\iota)$.
Given a Boolean formula~$F$ and an assignment~$\iota$ for~$\var(F)$.
Then, for~$F$ in CNF, $F[\iota]$ is a Boolean formula obtained by removing every~$c\in F$
with~$x\in c$ and $\neg x\in c$ if %
$\iota(x)=1$ and $\iota(x)=0$, respectively,
and by removing from every remaining clause~$c\in F$ literals~$x$ and $\neg x$
with~$\iota(x)=0$ and $\iota(x)=1$, respectively.
Analogously, for~$F$ in DNF values $0$ and~$1$ are swapped. %
  For a given QBF~$\Q$ and an assignment~$\iota$, $\Q[\iota]$ is a
  QBF obtained from~$\Q$, where variables~$x\in\dom(\iota)$ are removed from preceding
  quantifiers accordingly, and~$\matr(\Q[\iota])\eqdef (\matr(\Q))[\iota]$.
A Boolean formula~$F$ \emph{evaluates to true}
if there exists an assignment~$\iota$ for~$\var(F)$
such that~$F[\iota]=\emptyset$ if~$F$ is in CNF or~$F[\iota]=\{\emptyset\}$
if~$F$ is in DNF. 
  A closed QBF~$\Q$ \emph{evaluates to true (or is valid)} if~$\ell=0$
  and the Boolean formula $\matr(Q)$ evaluates to true.}
Otherwise,
  i.e., if~$\ell \neq 0$, we distinguish according to~$Q_1$.
  If~$Q_1=\exists$, then~$\Q$ evaluates to true if and only if there
  exists an assignment~$\iota: V_1\rightarrow \{0,1\}$ such
  that~$\Q[\iota]$ evaluates to true.  If~$Q_1=\forall$,
  then~$\Q[\iota]$ evaluates to true if for any
  assignment~$\iota: V_1 \rightarrow\{0,1\}$, $\Q[\iota]$ evaluates to
  true.  An \FIX{(open or closed)} QBF~$\Q$ is \emph{satisfiable} if
  there is a truth assignment~$\iota: \fvar(\Q) \rightarrow \{0,1\}$
  such that resulting closed QBF~$\Q[\iota]$ evaluates to
  true. Otherwise~$\Q$ is \emph{unsatisfiable}.
\FIX{Given a closed QBF~$\Q$, the \emph{evaluation problem~\QBFSAT} of QBFs asks
  whether $\Q$ evaluates to true; $\ell$\hy\QBFSAT refers to the problem~\QBFSAT on QBFs of quantifier rank~$\ell$.}
  The problem~\QBFSAT is \PSPACE-complete and is therefore believed to
  be computationally harder than
  \SAT~\cite{KleineBuningLettman99,Papadimitriou94,StockmeyerMeyer73}.
  \FIXCAM{For more details on QBFs we refer to %
\cite{BiereHeuleMaarenWalsh09,KleineBuningLettman99}.}

  The \emph{projected model counting problem}~$\PQBF$ takes an open
  QBF~$\Q$ and asks to output the number of distinct
  assignments~$\iota: \fvar(\Q) \rightarrow \{0,1\}$ such
  that~$\Q[\iota]$ evaluates to true.

\begin{EX}\label{ex:running1}
  Consider the closed QBF~$\Q=\exists w,x. \forall y,z. D$,
  where~$D\eqdef d_1 \vee d_2 \vee d_3 \vee d_4$,
  and~$d_1\eqdef w \wedge x \wedge \neg y$,
  $d_2\eqdef \neg w \wedge \neg x \wedge y$,
  $d_3\eqdef w \wedge y \wedge \neg z$, and
  $d_4\eqdef w \wedge y \wedge z$.  
  Observe that~$\Q[\iota]$ is valid under
  assignment~$\iota = \{w \mapsto 1, x \mapsto 1\}$.
  In particular, $\Q[\iota]$ can be simplified
  to~$\forall y,z. (\neg y) \vee (y \wedge \neg z) \vee (y \wedge z)$,
  which is valid, since for any
  assignment~$\kappa: \{y,z\} \rightarrow\{0,1\}$
  {%
    the formula~$Q[\iota][\kappa]$ (and therefore~$Q$) evaluates to true.
  }%
\end{EX}

\smallskip
\noindent\textbf{Tree Decompositions (TDs). }%
For basic terminology on graphs and digraphs, we refer to standard
texts~\cite{Diestel12,BondyMurty08}.  \FIXCAM{For an \emph{arborescence}~$T=(N,A,r)$,
which is a directed, rooted tree with root~$r$ and a node~$t \in N$, 
we let $\children(t, T)$ be the
set of all \emph{child nodes}~$t'$, which have an outgoing edge~$(t,t') \in A$ from~$t$ to~$t'$.}
Let $G=(V,E)$ be a graph.
\FIXCAM{A \emph{tree decomposition (TD)} of graph~$G$ is a pair
$\TTT=(T,\chi)$ where $T=(N,A,r)$ is an arborescence with root $r\in N$,
and $\chi$ is} a mapping that assigns to each node $t\in N$ a set
$\chi(t)\subseteq V$, called a \emph{bag}, such that the following
conditions hold:
(i) $V=\bigcup_{t\in N}\chi(t)$ and
$E \subseteq\bigcup_{t\in N}\SB \{u,v\} \SM u,v\in \chi(t)\SE$; and (ii)
for each \FIXCAM{$q, s, t$,} such that $s$ lies on the path from $q$ to
$t$, we have $\chi(q) \cap \chi(t) \subseteq \chi(s)$.
Then, $\width(\TTT) \eqdef \max_{t\in N}\Card{\chi(t)}-1$.  The
\emph{treewidth} $\tw{G}$ of $G$ is the minimum $\width({\TTT})$ over
all tree decompositions $\TTT$ of $G$.
For arbitrary but fixed $w \geq 1$, it is feasible in linear time to
decide if a graph has treewidth at most~$w$ and, if so, to compute a
tree decomposition of width $w$~\cite{Bodlaender96}.
\FIX{Further, we call a tree decomposition~$\mathcal{T}=(T,\chi)$ a \emph{path decomposition (PD)} if~$T=(N,\cdot,r)$ 
and~$\Card{\children(t)}\leq 1$ for each node~$t\in N$. %
\FIX{Analogously, we} define~\emph{pathwidth} $\pw{G}$ as the minimum~$\width(\mathcal{T})$
over all path decompositions of~$G$.}
Similarly, for~$m\geq 2$, let \emph{$m$-local pathwidth} of~$G$ refer to the pathwidth
over all path decompositions of~$G$, where \FIXCAM{each vertex} in~$V$ occurs
in at most~$m$ bags.
For a given tree decomposition~$\mathcal{T}=(T,\chi)$
with~$T=(N,A,r)$, and an element~$x\in\bigcup_{t \in N}\chi(t)$,
we denote by~$\mathcal{T}[x]$ the result~$\mathcal{T}'$ 
of restricting~$\mathcal{T}$ to nodes, whose 
bags contain~$x$.
Formally, $\mathcal{T}'\eqdef (T',\chi')$, where~$T'\eqdef (N',A',r')$,
$N'\eqdef\{t\mid t\in N, x\in\chi(t)\}$, $A'\eqdef A\cap(N'\times N')$, and for each $t\in N'$, $\chi'(t)=\chi(t)$.
Finally, $r'\in N'$ is the first node reachable from $r$.
\begin{figure}[t]%
\centering
\includegraphics[scale=.793]{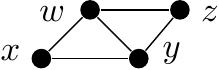}
\includegraphics[scale=.793]{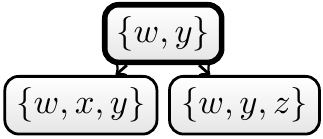}
\includegraphics[scale=.793]{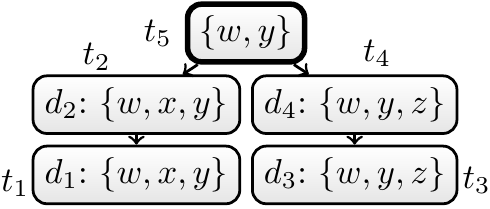}
\caption{Primal graph~$P_{\Q}$ of~$\Q$ from Example~\ref{ex:running1}
  (left) with TDs~${\mathcal{T}}_1, {\mathcal{T}}_2$ of graph~$P_{\Q}$
  (right).}%
\label{fig:graph-td}%
\end{figure}
The literature distinguishes so-called nice tree decompositions, which can be
computed in linear time without increasing the width~\cite{Kloks94a}.
For our purposes, the \FIXCAM{following} relaxed variant of almost nice tree decompositions \FIXCAM{suffices}. %
\begin{DEF}
  Given an integer~$c \in \NAT$. A tree
  decomposition~$\mathcal{T}=(T, \chi)$, where~$T=(N,\cdot,r)$, is
  called \emph{almost $c$-nice}, if for each node~$t\in N$ with
  $\children(t)=\{ t_1, \ldots, t_s\}$, the following conditions
  are true (i) $s\leq 2$ and
  (ii)~$\Card{\chi(t) \setminus \bigcup_{i=1}^{i=s}\chi(t_i)} \leq c$.
\end{DEF}

\trash{Almost nice tree decompositions are a true relaxation of nice tree
decompositions and can similarly be computed in linear time without
increasing the width.}

In order to use tree decompositions for QBFs, we need a
graph representation of Boolean formulas~\cite{SamerSzeider10b}.
The \emph{primal graph}~$P_F$ of a Boolean formula~$F$ in CNF or DNF
has the variables~$\var(F)$ of~$F$ as vertices and an edge~$\{x,y\}$
if there exists a term or clause~$f \in F$ such that
$x,y \in \var(f)$, respectively.  For a QBF~$\Q$, {we identify
  its primal graph with the primal graph of its matrix, i.e., let}
$P_{\Q}\eqdef P_{\matr(\Q)}$.

\begin{EX}
  Figure~\ref{fig:graph-td} illustrates the primal graph~$P_\Q$
  {of the QBF from Example~\ref{ex:running1}} and two tree
  decompositions of~$P_\Q$ of width~$2$. 
  The graph $P_\Q$ has treewidth~$2$, since the vertices $w$,$x$,$y$
  \FIXCAM{are completely connected and hence width~$2$ is optimal~\cite{Kloks94a}}.
\end{EX}

\begin{DEF}\label{def:labeledTD}
  Let $\mathcal{T}=(T,\chi)$ be a tree decomposition of a graph~$G$.
  A \emph{labeled tree decomposition (LTD)}~$\TTT$ of a Boolean
  formula~$F$ in CNF or DNF is a tuple~$\TTT=(T,\chi,\delta)$ where
  $(T,\chi)$ is a tree decomposition of $P_F$ with~$T=(N,\cdot,r)$,
  and $\delta: N' \rightarrow {F}$ with~$N'\subseteq N$ is a 
\FIXCAM{bijective mapping from TD nodes to clauses or terms in~$F$ such that for every~$t\in N'$,
  $\var(\delta(t))\subseteq \chi(t)$.}
\end{DEF}

\begin{OBS}\label{obs:uniqueremoval}
  Given an almost $c$-nice tree decomposition~$\mathcal{T}$ of a
  primal graph~$P_F$ for a given 3-CNF or 3-DNF formula~$F$ of
  width~$k$. Then, one can easily create a labeled almost $c$-nice
  tree decomposition~$\mathcal{T'}$ of~$P_F$ with a linear number of
  nodes in the number of nodes in~$\mathcal{T}$ such
  that~$\width(\mathcal{T}')=\width(\mathcal{T})$.
\end{OBS}
\begin{EX}
Consider again Figure~\ref{fig:graph-td} (right). Observe that %
$\mathcal{T}_2$ is a labeled almost $3$-nice tree decomposition of~$P_\Q$,
\FIXCAM{where
labeling function~$\delta$ sets~$\delta(t_i)=d_i$, for~$1\leq i\leq 4$.}
\end{EX}

\section{Decomposition-Guided Compression}
\footnoteitext{\label{foot:tprimenotdefinedorisit}For the sake of readability, we defer the discussion of formal details on the construction of~$\mathcal{T}'$ to
the proof of Lemma~\ref{lem:compr}.}%
\begin{figure}[t]%
\centering%
\includegraphics[scale=0.85]{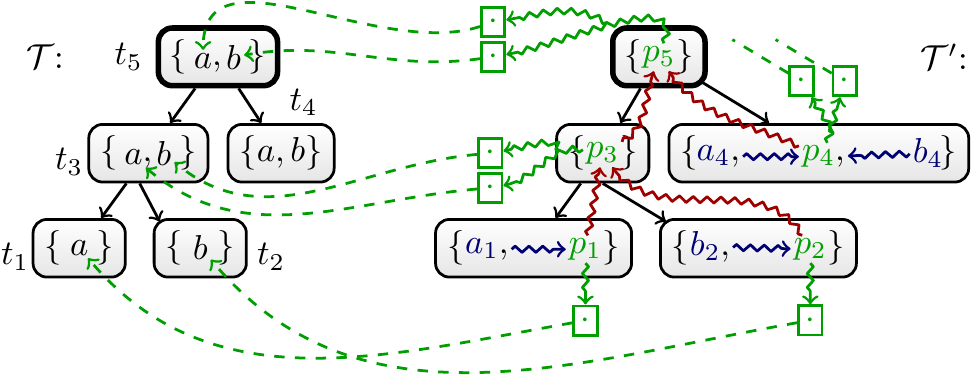}
\caption{Simplified illustration of a certain
tree decomposition $\mathcal{T}'=(T,\chi')$ of~$P_{R(\Q)}$ (yielded$^{\ref{foot:tprimenotdefinedorisit}}$
by reduction~$R$), and its relation to
tree decomposition~$\mathcal{T}=(T,\chi)$ of~$P_\Q$.
Each bag~$\chi'(t_i)$ of a node~$t_i$ of $\mathcal{T}'$ contains variable~$x_i$ for any variable~$x$ introduced in~$\chi(t_i)$
and $\ceil{\log(\width(\mathcal{T}))}$ many (green) pointer variables~$p_i$ selecting \emph{one} variable in~$\chi(t_i)$ of~$\mathcal{T}$.
Squiggly red arrows indicate the propagation between pointers~$p_i, p_j$ and ensure consistency.
In particular, although truth values for variable~$a$ are ``guessed'' using~$a_1$ and~$a_4$
(and ``propagated'' via blue squiggly arrows to corresponding pointers~$p_1$ and $p_4$, respectively),
these red arrows ensure via pointers~$p_1,p_3,p_4,p_5$ that truth values for~$a_1$ and~$a_4$ %
coincide.}
\label{fig:transformation}%
\end{figure}
Next, we present our %
approach %
to transform a given input instance of treewidth $k$ into an
instance of exponentially smaller treewidth (``compression'')
\FIX{compared to the original treewidth $k$}.  Thereby, we trade the
compression of the parameter for the cost of additional computation
power required to solve the compressed instance.
For the canonical $\QBFSAT$ problem, we require  an increased
quantifier~rank.

First, we introduce the reduction~$R$ that takes an instance~$\Q$ of
$\ell\hy\QBFSAT$, and computes a corresponding tree
decomposition~$\mathcal{T}$ of the primal graph~$P_\Q$,
where~$\width(\mathcal{T})=\tw{P_\Q}$.  Then, it returns a compressed
instance~$R(\Q)$ of~$(\ell+1)\hy\QBFSAT$ of
treewidth~$\bigO(\log(\width(\mathcal{T})))$.  The reduction~$R$,
which is guided by TD~$\mathcal{T}$,
yields$^{\ref{foot:tprimenotdefinedorisit}}$ a new \emph{compressed
  tree decomposition}~$\mathcal{T}'$ of~$P_{R(\Q)}$ of
width~$\bigO(\log(\width(\mathcal{T})))$.
For that it is crucial to \FIXCAM{balance}
introducing copies of variables (redundancy) and saving treewidth
(structural dependency), such that, intuitively, we can still
evaluate~$R(\Q)$ given the limitation of
treewidth~$\bigO(\log(\width(\mathcal{T})))$.
\FIXCAM{To keep this balance, we can only analyze %
in a bag in
$\mathcal{T}'$ a \emph{constant} number of elements of the corresponding original bag
of~$\mathcal{T}$.}
Still, {considering} $\log(\width(\mathcal{T}))$ many %
elements in a bag at \FIXCAM{once allows} us to represent one ``pointer'' to
address at most~$\width(\mathcal{T})$ many elements of each bag
of~$\mathcal{T}$ and, consequently, the restriction to
$\bigO(\log(\width(\mathcal{T})))$ many elements in a bag at once
enables \emph{constantly} many such pointers.
To give a first glance at the idea of the reduction~$R$,
Figure~\ref{fig:transformation} provides an intuition and
illustrates a tree decomposition~$\mathcal{T}$
of~$P_\Q$ together with a corresponding compressed tree
decomposition~$\mathcal{T}'$ of~$P_{R(\Q)}$, whose bags contain
pointers to original bags of~$\mathcal{T}$.  Actually, we can encode
the \emph{propagation} of information from one bag of~$\mathcal{T}'$
to its parent bag with the help of these pointers. Thereby, we ensure
that information is consistent and
this consistency can be preserved, even though we \emph{guess} in~$R$ truth values
for copies of the same variable in~$\Q$ {independently}.
\FIX{Note that these ``local'' pointers for each bag 
are essential to achieve treewidth compression.}

Below, we discuss \FIXCAM{the} reduction~$R$ in more detail.
Then, Section~\ref{sec:metho} provides a description of a general
methodology for establishing lower bounds for problems parameterized
by treewidth.  In more detail, equipped with our lower bound results
for~\QBFSAT, we propose reductions from~\QBFSAT as a general toolkit
for proving lower bounds assuming ETH.  Further, we discuss several
showcases to illustrate %
\FIXCAM{this methodology}.

\subsection{The Reduction}\label{sec:reduction}

The formula~$R(\Q)$ constructed by~$R$ mainly consists of three interacting parts.
In the presentation, we  \FIXCAM{refer to} them as \emph{guess}, \emph{check},  and \emph{propagate} part.
\begin{itemize}[leftmargin=0pt, itemindent=10pt, itemsep=5pt]
	\item "Guess"~($\mathcal{G}$): Contains clauses responsible for guessing truth values of variables occurring in the original QBF~$\Q$.
	\item "Check"~($\mathcal{CK}$): These clauses ensure that there is at least one 3-DNF term in~$\Q$ that is satisfied, thereby maintaining 3 pointers for each node as discussed above.
	\item "Propagate"~($\mathcal{P}$): These clauses ensure consistency using a pointer for each node of the tree decomposition. %
\end{itemize}%
We commence with the formal description of~$R$.
Given a QBF~$\Q$ of the form~$\Q\eqdef Q_1 V_1. {{Q_2}} V_2. \cdots \forall V_\ell. D$, where~$D$ is in 3-DNF
such that the quantifier blocks are alternating,~i.e., quantifiers of quantifier blocks with even indices are equal, which are
different from those of blocks with odd indices.
Further, assume a labeled almost $c$-nice tree decomposition~$\mathcal{T} = (T, \chi, \delta)$, where~$T=(N, \cdot, \cdot)$ of the primal graph~$P_\Q$ of~$D$,
which always exists by Observation~\ref{obs:uniqueremoval}.
Notice that by Definition~\ref{def:labeledTD} for all terms~$d\in D$, the inverse function $\delta^{-1}(d)$ is well-defined. %
\FIX{Further, actually $R$ can deal with open QBFs, i.e, QBF $\Q$ does not necessarily have to be closed. Open formulas are needed later to simplify the correctness proof~of~Section~\ref{sec:correctness}.}

We use the following sets of variables.
Let~$\nint{\mathcal{T}}{x}{:=}$ $\{t \mid t \in N, x
\in \chi(t)\setminus({\mathsmaller\bigcup_{t_i\in\children(t)}\chi(t_i)})\}$ be the
set of nodes, where a given element~$x$ is
\emph{introduced}. %
For a set~$V\subseteq\var(D)$ of variables, we denote
by~$\lint{\mathcal{T}}{V} \eqdef \{x_t \mid x \in V, t \in
\nint{\mathcal{T}}{x}\}$ the set of fresh variables generated for each
original variable~$x$ and node~$t$, where~$x$ is introduced.  Later, we
need to distinguish whether the set~$V_i$ of variables is universally
or existentially quantified.  Universal quantification requires to
shift for each~$x\in V_i$ all but one representative
of~$\{x_t \mid t \in \nint{\mathcal{T}}{x}\}$ to the next existential quantifier block $Q_{i+1}$.
The representative variable that is not shifted is denoted by~$\master{x}$.
In particular, given a quantifier block~$Q_2$, its variables~$V_2$ and
the variables~$V_1$ of the preceding quantifier block, we define:
\FIXCAM{$\lint{\mathcal{T}}{Q_2,V_2,V_1} \eqdef \{x_t \mid x \in V_2, t \in
\nint{\mathcal{T}}{x}, Q_2=\exists\} \cup \{x_t \mid x \in V_2, t \in\nint{\mathcal{T}}{x},
x_t = \master{x}, Q_2 = \forall\} \cup \{x_t \mid x \in
V_1, t \in \nint{\mathcal{T}}{x}, x_t \neq \master{x}, Q_2{=}$ $\exists\}$}.  We denote
by~$\lsat \eqdef \{sat_t, sat_{\leq t} \mid t \in N\}$ the set of
fresh decision variables responsible for storing for each
node~$t\in N$ whether any term at $t$ or at any node below~$t$ is satisfied, respectively.
Finally, we denote by~$\lbv{\mathcal{T}} \eqdef \{b_t^0, \ldots, {\mathsmaller b_t^{\ceil{\log(\Card{\chi(t)})}-1}} \mid t\in N\}$,
and~$\lbvv{\mathcal{T}} \eqdef \{v_t \mid t\in N\}$ the set of fresh variables
for each node~$t\in N$ that will be used to address particular elements of the corresponding bags (pointer as depicted in Figure~\ref{fig:transformation} in binary representation), and to assign truth values for these elements, respectively.
Overall, the variables in~$\lbv{\mathcal{T}}$ allow us to guide the evaluation of formula~$D$ along the tree decomposition~$\mathcal{T}$.
For checking 3-DNF terms, we need the same functionality three more times,
resulting in the sets~$\lbvd{\mathcal{T}} \eqdef \{b_{t,j}^0, \ldots, {\mathsmaller b_{t,j}^{\ceil{\log(\Card{\chi(t)} + 1)}-1}} \mid t\in N, 1 \leq j \leq 3\}$ that additionally may refer to a special fresh element~$nil$ (therefore the $+1$ in the exponent in definition of~$\lbvd{\mathcal{T}}$),
and~$\lbvvd{\mathcal{T}} \eqdef \{v_{t,j} \mid t\in N, 1 \leq j \leq 3\}$ of fresh variables.
Notice that the construction is designed in such a way that the focus lies only on certain elements of the bag
(one at a time, and independent of other elements within the same bag).
In the end, this ensures that the treewidth of our reduced instance is only logarithmic in the original treewidth of the primal graph of~$D$.
Reduction~$R(\Q)$ creates~$\Q'\eqdef$
{\smallalign{\normalfont\small}
\begin{align*}
	&Q_1\ \lint{\mathcal{T}}{Q_1, V_1, \emptyset}.\ Q_2\ \lint{\mathcal{T}}{Q_2, V_2, V_1}.\cdots %
	\forall\ \lint{\mathcal{T}}{\forall, V_\ell, V_{\ell-1}}, \lbv{\mathcal{T}}.\notag\\
	&\exists\ \lint{\mathcal{T}}{\exists, \emptyset, V_{\ell}},\ \lbvv{\mathcal{T}}, \lbvvd{\mathcal{T}}, \lbvd{\mathcal{T}}, \lsat.\ C,%
\end{align*}}%
\FIXCAM{where~$C$ is a CNF formula consisting of 
guess, check and~propagate parts, i.e., sets~$\mathcal{G}, \mathcal{CK}$, and~$\mathcal{P}$ of~clauses, respectively.} %
\begin{EX}\label{ex:red}
Consider again~$\Q$ from Example~\ref{ex:running1}.
The resulting instance~$R(\Q)$ looks as follows assuming that~$\master{y}=y_{t_1}$, where~$C$ consists of a guess, check and, propagate part.
{\smallalign{\footnotesize}
	\begin{align*}
		&\exists \underbrace{w_{t_1}, w_{t_3}, x_{t_1}.}_{\lint{\mathcal{T}}{\exists, \{w,x\},\emptyset}}\, \forall\hspace{-.5em} 
		\underbrace{y_{t_1}, z_{t_3},}_{\lint{\mathcal{T}}{\forall, \{y,z\},\{w,x\}}} 
		\hspace{-1em}\underbrace{b_{t_1}^0, b_{t_1}^1, b_{t_2}^0, \ldots,  b_{t_4}^1, b_{t_5}^0.}_{\lbv{\mathcal{T}}}\, \exists 
		\hspace{-.5em}\underbrace{y_{t_3},}_{\lint{\mathcal{T}}{\exists, \emptyset, \{y,z\}}}\\
		&%
		\underbrace{v_{t_1}, \ldots, v_{t_5}}_{\lbvv{\mathcal{T}}},%
		\underbrace{v_{t_1, 1}, v_{t_1, 2}, v_{t_1, 3}, v_{t_2, 1}, \ldots, v_{t_5,3},}_\lbvvd{\mathcal{T}}\\
		&\underbrace{b_{t_1, 1}^0, b_{t_1, 1}^1, b_{t_1, 2}^0, \ldots, b_{t_5, 3}^1,}_{\lbvd{\mathcal{T}}} %
		\underbrace{sat_{t_1}, \ldots, sat_{t_5}, sat_{\leq t_1}, \ldots, sat_{\leq t_5}.}_\lsat C\\[-2.5em]
	\end{align*}%
}%
\end{EX}%
In the following, we define
sets~${\mathcal{G}}$, ${\mathcal{CK}}$, and ${\mathcal{P}}$ of clauses.
To this end, we require for the pointers a bit-vector (binary) representation of the elements in a bag of~$\mathcal{T}$,
and a mapping that assigns bag elements to its corresponding binary representation.
In particular, we assume an arbitrary, but fixed total order~$\prec$ of elements of a bag~$\chi(t)$ of any given node~$t\in N$. %
With~$\prec$, we can then assign each element~$x$ in~$\chi(t)$ its unique (within the bag) induced ordinal number~$o(t,x)$.
This ordinal number~$o(t,x)$ is expressed in binary.
For that we need precisely~$\ceil{\log(\Card{\chi(t)})}$ many
bit-variables~$B\eqdef \{ b_t^0, \ldots, {\mathsmaller b_t^{\ceil{\log(\Card{\chi(t)})}-1}}\}$.
We denote by~$\bval{x}{t}$ the (consistent) set of literals over variables in~$B$ that encode (in binary)
the ordinal number~$o(t,x)$ of~$x\in\chi(t)$ in~$t$, such that whenever a literal~$b_t^i$ or~$\neg b_t^i$ is contained in the set~$\bval{x}{t}$, the $i$-th bit in the unique binary representation of~$o(t,x)$ is 1 or 0, respectively.
Analogously, for~$1\leq j\leq 3$ we denote by~$\bvali{x}{t}{j}$ the (consistent) set of literals over variables in~$B_{j}\eqdef \{ b_{t,j}^0, \ldots, {\mathsmaller b_{t,j}^{\ceil{\log(\Card{\chi(t)} + 1)}-1}}\}$ that \FIXCAM{either} binary-encode \FIXCAM{the} ordinal number~$o(t,x)$ of~$x\in\chi(t)$ in~$t$, \FIXCAM{or these literals binary-encode number~$\text{max}_{y\in\chi(t)}o(t,y)+1$ for~$x=nil$}.

\medskip
\noindent\textbf{The guess part~$\mathcal{G}$.} The clauses in~$\mathcal{G}$, which we denote as implications, are defined as follows.
{
\smallalign{\normalfont\small}
\begin{align}
	\label{red:guessatom}&x_t \wedge \bigwedge_{b \in \bval{x}{t}} b \longrightarrow v_t&\pushright{\text{for each } x_t\in \lint{\mathcal{T}}{\var(D)}}\\
	\label{red:guessnegatom}\neg &x_t \wedge \bigwedge_{b \in \bval{x}{t}} b \longrightarrow \neg v_t&\pushright{\text{for each } x_t\in \lint{\mathcal{T}}{\var(D)}}%
\end{align}
}

\noindent Intuitively, this establishes that whenever a certain variable~$x_t$ for an introduced variable~$x\in\chi(t)$ is assigned to true (false)
and all the corresponding literals in~$\bval{x}{t}$ of the binary representation of~$o(t,x)$ are satisfied (i.e, $x$ is ``selected''),
then also~$v_t\in\lbvv{\mathcal{T}}$ of node~$t$ has to be set to true (false).

Analogously, %
set~$\mathcal{G}$ further contains the following clauses: %
{
\vspace{-.5em}
\smallalign{\normalfont\small}
\begin{align}
	\label{red:guessatomj}&x_t \wedge \bigwedge_{b \in \bvali{x}{t}{j}} b \longrightarrow v_{t,j}&\pushright{\text{for each } x_t\in \lint{\mathcal{T}}{\var(D)}, 1 \leq j \leq 3}\raisetag{1.15em}\\
	\label{red:guessnegatomj}\neg &x_t \wedge \bigwedge_{b \in \bvali{x}{t}{j}} b \longrightarrow \neg v_{t,j}&\pushright{\text{for each } x_t\in \lint{\mathcal{T}}{\var(D)}, 1 \leq j \leq 3}\raisetag{1.15em}%
\end{align}
}
\begin{EX}
Consider formula~$C$ from Example~\ref{ex:red}.
Let~$1\leq j\leq 3$. Further, assume the following mapping of bag contents to bit-vector assignments.
For any variable~$a\in\var(D)$ with $t\in\nint{\mathcal{T}}{a}$ and for~$a=nil$ with~$t\in N$,
we arbitrarily fix the total ordering~$\prec$ and have $\bval{a}{t}$ and $\bvali{a}{t}{j}$ as follows. %
{%
	\begin{table}[h]%
		\shortversion{\footnotesize}
		\fontsize{7.7}{8}\selectfont
		\longversion\centering
		\vspace{-.75em}%
		\begin{tabular}[h]{@{\hspace{0em}}r@{\hspace{.15em}}|@{\hspace{.1em}}c@{\hspace{.1em}}|@{\hspace{.1em}}c@{\hspace{.1em}}|@{\hspace{.1em}}c@{\hspace{.1em}}|@{\hspace{.1em}}c@{\hspace{.1em}}|@{\hspace{.1em}}c@{\hspace{.1em}}|@{\hspace{.1em}}c@{\hspace{0em}}}%
			\multirow{2}{*}{$a$} &
			\multicolumn{2}{c}{$t\in\{t_1,t_2\}$} & \multicolumn{2}{c}{$t\in\{t_3,t_4\}$} & \multicolumn{2}{c}{$t=t_5$}\\
			& \bval{a}{t} & \bvali{a}{t}{j} & \bval{a}{t} & \bvali{a}{t}{j} & \bval{a}{t} & \bvali{a}{t}{j} \\\hline
			{\scriptsize$w$} & $\{\neg b_{t}^0, \neg b_{t}^1\}$ & $\{\neg b_{t,j}^0, \neg b_{t,j}^1\}$ & $\{\neg b_{t}^0, \neg b_{t}^1\}$ & $\{\neg b_{t,j}^0, \neg b_{t,j}^1\}$ & $\{\neg b_{t}^0\}$ & $\{\neg b_{t,j}^0, \neg b_{t,j}^1\}$ \\
			{\scriptsize$x$} & $\{\neg b_{t}^0, b_{t}^1\}$ & $\{\neg b_{t,j}^0, b_{t,j}^1\}$ & \multicolumn{2}{c|@{\hspace{.1em}}}{-}& \multicolumn{2}{c}{-} \\
			{\scriptsize$y$} & $\{b_{t}^0, \neg b_{t}^1\}$ & $\{b_{t,j}^0, \neg b_{t,j}^1\}$ & $\{\neg b_{t}^0, b_{t}^1\}$ & $\{\neg b_{t,j}^0, b_{t,j}^1\}$ & $\{b_{t}^0\}$ & $\{\neg b_{t,j}^0, b_{t,j}^1\}$ \\
			{\scriptsize$z$} & \multicolumn{2}{c|@{\hspace{.1em}}}{-} & $\{b_{t}^0, \neg b_{t}^1\}$ & $\{b_{t,j}^0, \neg b_{t,j}^1\}$ & \multicolumn{2}{c}{-} \\
			{\scriptsize$nil$} & - & $\{b_{t,j}^0, b_{t,j}^1\}$ & - & $\{b_{t,j}^0, b_{t,j}^1\}$ & - &  $\{b_{t,j}^0, \neg b_{t,j}^1\}$
		\end{tabular}\vspace{-.75em}%
	\end{table}%
}

\noindent The guess part of~$C$ contains for example for variable~$w\in\var(D)$ the following clauses.%
{
\smallalign{\small}
	\begin{align*}%
		w_{t_1} \wedge \neg b_{t_1}^0 \wedge \neg b_{t_1}^1 \longrightarrow v_{t_1},\quad %
		\neg w_{t_1} \wedge \neg b_{t_1}^0 \wedge \neg b_{t_1}^1 \longrightarrow \neg v_{t_1},\end{align*}\begin{align*} %
		w_{t_3} \wedge \neg b_{t_3}^0 \wedge \neg b_{t_3}^1 \longrightarrow v_{t_3}, \quad
		\neg w_{t_3} \wedge \neg b_{t_3}^0 \wedge \neg b_{t_3}^1 \longrightarrow \neg v_{t_3}.
	\end{align*}%
}%
Thereby, whenever we guess a certain truth value for~$w_{t_1}$ ($w_{t_3}$) it is ensured that there is a certain bit-vector, namely $\bval{w}{t_1}$ ($\bval{w}{t_3}$) such that~$v_{t_1}$ ($v_{t_3}$) has to be set to the same truth value.
Analogously, clauses of the form~(\ref{red:guessatomj}) and~(\ref{red:guessnegatomj}) are 
in~$\mathcal{G}$.
\end{EX}

\medskip
\noindent\textbf{The check part~$\mathcal{CK}$.} In the following, we assume an arbitrary, but fixed total order of the (three) literals of each (3-DNF) term~$d\in D$.
We refer to the first, second, and third literal of~$d$ by~$\lit{d}{1}, \lit{d}{2}$, and $\lit{d}{3}$, respectively.
Analogously, $\ato{d}{1}, \ato{d}{2},$ and~$\ato{d}{3}$
refers to the variable of the first, second, and third literal, respectively. 
Further, for a given term~$d\in D$ and~$1 \leq j \leq 3$,
let~$\bvv{d}{t}{j}$ denote $v_{t,j}$ if~$\lit{d}{j}$ is a variable, and~$\neg v_{t,j}$ otherwise.
Set~$\mathcal{CK}$ contains the following clauses:
{
\smallalign{\normalfont\small}
\begin{align}
	\label{chk_leqtright}&sat_{\leq t} \longrightarrow sat_{\leq t_1} \vee \cdots \vee sat_{\leq t_s} \vee sat_t&\pushleft{\text{for each } t\in N,}\notag\vspace{-0.4em}\\
	&&\pushleft{\children(t)=\{ t_1, \ldots, t_s\}} %
\end{align}
}%
Informally speaking, for any node~$t$ this ensures the propagation of whether
we satisfied at least one term directly in node~$t$, or in any descendant of~$t$.

In order to check whether a particular term is satisfied, we add for each term~$d\in D$ clauses encoding the implication
$sat_{\lnode{d}} \longrightarrow \bigwedge_{1 \leq j \leq 3}\hspace{-.3em}\left[\bigwedge_{b\in \bvali{\ato{d}{j}}{\lnode{d}}{j}} b \wedge \bvv{d}{\lnode{d}}{j}\right]$ as follows:
{
\smallalign{\normalfont\small}
\begin{align}
	\label{chk_satcb}&\pushleft{sat_{\lnode{d}} \longrightarrow b}&\pushleft{\text{for each } d\in D, 1\leq j\leq 3,}\notag\vspace{-0.4em}\\
	&&\pushleft{b\in \bvali{\ato{d}{j}}{\lnode{d}}{j}}\\
	\label{chk_satcbv}&sat_{\lnode{d}} \longrightarrow \bvv{d}{\lnode{d}}{j}&\pushleft{\text{for each } d\in D, 1\leq j\leq 3}%
\end{align}
}%

\noindent Finally, we add~$sat_{\leq r}$ for root~$r$, and~$\neg sat_t$ for each node~$t$ in~$N\setminus\bigcup_{d\in D}\{\lnode{d}\}$
\FIXCAM{since these nodes are not used for checking satisfiability of any term.} 
{
\smallalign{\normalfont\small}
\begin{align}
	\label{chk_satn}&sat_{\leq r}\\
	\label{chk_negsatt}&\neg sat_t&\pushright{\text{for each } t\in N\setminus\bigcup_{d\in D}\{\lnode{d}\}} %
\end{align}
}
\begin{EX}
	Consider again formula~$C$ from Example~\ref{ex:red}.
	We discuss clauses of the check part for node~$t_2=\lnode{d_2}$ and root node~$t_5$. Thereby, we encode satisfiability of term~$d_2=\neg w \wedge \neg x \wedge y$ assuming~$\lit{d_2}{1}=\neg w, \lit{d_2}{2}=\neg x$, and $\lit{d_2}{3}=y$.
\shortversion{\vspace{-1.25em}}
{
\smallalign{\small}
	\begin{align*}
		sat_{\leq {t_2}} \longrightarrow sat_{\leq t_1} \vee sat_{t_2},\\ %
		sat_{t_2} \longrightarrow \neg b_{t_2,1}^0, \quad sat_{t_2} \longrightarrow  \neg b_{t_2,1}^1, \quad sat_{t_2} \longrightarrow \neg v_{t_2,1},\\
		sat_{t_2} \longrightarrow \neg b_{t_2,2}^0, \quad  sat_{t_2} \longrightarrow b_{t_2,2}^1, \quad sat_{t_2} \longrightarrow \neg v_{t_2,2},\\
		sat_{t_2} \longrightarrow b_{t_2,3}^0, \quad  sat_{t_2} \longrightarrow \neg b_{t_2,3}^1, \quad  sat_{t_2} \longrightarrow v_{t_2,3}, \\\\[-1em]
		 sat_{\leq t_5} \longrightarrow sat_{\leq t_2} \vee sat_{\leq t_4} \vee sat_{t_5}, \quad %
     sat_{\leq t_5},\quad \neg sat_{t_5}\\[-3em]
	\end{align*}%
}%
\end{EX}

\noindent\textbf{The propagate part~$\mathcal{P}$.} The sets~$\mathcal{G}$ and~$\mathcal{CK}$ contain clauses responsible for guessing truth values and checking that at least one term of the original formula~$D$ is satisfied accordingly.
In particular, the guess of truth values for $\var(D)$ happens at different tree decomposition nodes ``independently'', whereas checking whether at least one term~$d\in D$ is satisfied is achieved
in exactly one tree decomposition node~$\lnode{d}$.
Intuitively, in order to ensure that these independent guesses of truth values for~$\var(D)$,
are consistent, clauses in~$\mathcal{P}$ make use of the connectedness condition of TDs 
in order to guide the comparison of these independent guesses along the TD.
More precisely, for each tree decomposition node~$t\in N$, every node~$t_i\in\children(t)$,
and every variable~$x\in \chi(t)\cap\chi(t_i)$ that both nodes~$t$ and~$t_i$ have in common,
the set~$\mathcal{P}$ contains clauses: %
{
\smallalign{\normalfont\small}
\begin{align}
	\label{red:propatom}&v_t \wedge \bigwedge_{b \in \bval{x}{t}} b \wedge \bigwedge_{b \in \bval{x}{t_i}} b \longrightarrow v_{t_i}&\pushleft{\text{for each } t\in N, t_i\in\children(t),} \notag\vspace{-1.5em}\\ &&\pushleft{x\in \chi(t)\cap\chi(t_i)}\raisetag{1.15em}%
\end{align}\begin{align}	
\label{red:negpropatom}\neg &v_t \wedge \bigwedge_{b \in \bval{x}{t}} b \wedge \bigwedge_{b \in \bval{x}{t_i}} b \longrightarrow \neg v_{t_i}&\pushleft{\text{for each } t\in N, t_i\in\children(t),} \notag\vspace{-1.5em}\\ &&\pushleft{x\in \chi(t)\cap\chi(t_i)}\raisetag{1.15em}%
\end{align}
}%
\noindent Further, for each clause~$d\in D$, every node~$t_i$ \FIXCAM{in} $\children(\lnode{d})$, and~$1 \leq j \leq 3$ such that~$\ato{d}{j}\in \chi(t_i)$,
set~$\mathcal{P}$ contains: %
{
\smallalign{\normalfont\small}
\begin{align}
	\label{red:propatomj}&\bigwedge_{b' \in \bvali{\ato{d}{j}}{t}{j}} b' \longrightarrow b&\pushleft{\text{for each } d\in D \text{ with } 1\leq j\leq 3, }\notag\vspace{-1.65em}\\
	&&\pushleft{t=\lnode{d}, t_i\in\children(t), \ato{d}{j}\in \chi(t_i), }\notag\vspace{-.2em}\\
	&&\pushleft{b \in \bvali{\ato{d}{j}}{t_i}{j}}\raisetag{1.15em}\\
	\label{red:propvright}&v_{t,j} \longleftrightarrow v_{t_i,j}&\pushleft{\text{for each } t\in N, t_i\in\children(t), 1 \leq j\leq 3}\raisetag{.4em}%
\end{align}
}%
\noindent Vaguely speaking, this construction ensures that
whenever a bag element (using~$\lbvd{\mathcal{T}}$) or a truth value (using~$\lbvvd{\mathcal{T}}$) is ``selected'' in node~$t$,
we also have to select the same (if exists) below in children of~$t$. %

\begin{EX}
	Consider once more~$C$ from Example~\ref{ex:red}.
	We illustrate the propagate part for node~$t_4=\lnode{d_4}$ and variable~$w$ assuming that~$w=\ato{d_4}{1}$. Observe that~$d_4=w \wedge y \wedge z$, and $w \in \chi(t_4)\cap \chi(t_3)$.
{
\smallalign{\small}
	\begin{align*}
		v_{t_4} \wedge \underbrace{\neg b_{t_4}^0 \wedge \neg b_{t_4}^1}_{\bval{w}{t_4}} \wedge \underbrace{\neg b_{t_3}^0 \wedge \neg b_{t_3}^1}_{\bval{w}{t_3}} \longrightarrow v_{t_3}, \quad\\
		\neg v_{t_4} \wedge \underbrace{\neg b_{t_4}^0 \wedge \neg b_{t_4}^1}_{\bval{w}{t_4}}\wedge \underbrace{\neg b_{t_3}^0 \wedge \neg b_{t_3}^1}_{\bval{w}{t_3}} \longrightarrow \neg v_{t_3},\\
		\neg b_{t_3,1}^0 \wedge \neg b_{t_3,1}^1 \longrightarrow \neg b_{t_4,1}^0,\quad \neg b_{t_3,1}^0 \wedge \neg b_{t_3,1}^1 \longrightarrow \neg b_{t_4,1}^1,\\
		v_{t_4,1} \longleftrightarrow v_{t_3,1}, \quad v_{t_4,2} \longleftrightarrow v_{t_3,2}, \quad v_{t_4,3} \longleftrightarrow v_{t_3,3}\\[-3em]
	\end{align*}
}%
\end{EX}

\begin{REM}
Recalling Figure~\ref{fig:transformation}, we would like to highlight the relation between elements of the figure and variables or clauses of reduction~$R$ introduced above.
Blue elements~$a_1, b_2, a_4, b_4$ represent ``introduce variables''~$\lint{\mathcal{T}}{\var(D)}$ and the blue squiggly arrows visualize the guess part~$\mathcal{G}$.
Green elements~$p_1, p_2, p_3,$ $p_4,p_5$ represent ``pointer variables''~$\lbv{\mathcal{T}}$ and~$\lbvd{\mathcal{T}}$ and the green squiggly arrows point to elements of tree decomposition~$\mathcal{T}$.
Finally, red squiggly arrows visualize the propagate part~$\mathcal{P}$. (The check part~$\mathcal{CK}$ is not explicitly visualized.)
\end{REM}

\noindent\textbf{Converting~$C$ to 3-CNF formula~$C'$.}  Observe that by similar arguments (cf.,~\cite{LampisMitsou17}) one can transform using an additional reduction~$R'$ the CNF formula~$C$ of the QBF~$R(\Q)$ into~3-CNF, resulting in~$\Q''=R'(R(\Q))$ such that~$\tw{P_{\Q''}} \leq \tw{P_{R(Q)}} + 2$.
To this end, one has to perform the following standard reduction (cf.,~\cite{LampisMitsou17}):
As long as there exists a clause~$c\in C$ consisting of more than~$3$ literals, we introduce a fresh existentially quantified variable~$v$, remove~$c$ from~$C$ and replace it with two new clauses. 
The first new clause contains~$v$ and two literals of~$c$,
while the second clause contains~$\neg v$ and the remaining literals of~$c$.
Note that this standard reduction~$R'$ does not affect satisfiability,
and can be done such that it causes only constant increase of the treewidth
(cf., Lemma~\ref{lem:compr} and \cite{LampisMitsou17}).
\FIXCAM{Observe that by construction the same argument actually holds for pathwidth.} %

\subsection{Methodology for Lower Bounds}\label{sec:metho}

The reduction discussed in the previous subsection allows us to establish \FIXCAM{our main result, which is the following theorem.} %

\begin{THM}[QBF lower bound]\label{lab:primqbflb}
  Given \FIXCAM{an arbitrary QBF} of the
  form~$Q=Q_1 V_1. Q_2 V_2. Q_3 V_3 \longversion\allowbreak \cdots Q_\ell V_\ell. F$
  where~$\ell\geq 1$, and $F$ is a 3-CNF formula
  (if~$Q_\ell=\exists$), or~$F$ is a 3-DNF formula
  (if~$Q_\ell=\forall$).  Then, unless ETH fails, $\Q$ cannot be
  solved in time~$\tower(\ell,o(k))\cdot \BUGFIX{\poly(\Card{\var(F)})}{2^{o(\Card{\var(F)})}}$,
  where~$k$ is the treewidth of the primal graph~$P_\Q$.
\end{THM}
In the following, we first use this theorem to establish a full methodology to obtain lower bound results for bounded treewidth and then provide a proof for the theorem in the next section.
The result \FIXCAM{for~$\ell=2$}~(cf., \cite{LampisMitsou17}) has already been applied as a strategy to show lower bound results for problems in artificial intelligence, as for example abstract argumentation, abduction, circumscription, and projected model counting, that are hard for the second level of the polynomial hierarchy when parameterized by treewidth~\cite{FichteHecher18c,FichteEtAl18,LampisMitsouMengel18}.
With the generalization to an \emph{arbitrary} quantifier rank in Theorem~\ref{lab:primqbflb}, one can obtain lower bounds for variants of these problems and even more general problems on the third level or higher levels of the polynomial hierarchy.

\smallskip%
\noindent\textbf{Methodology.}
\noindent This motivates \FIXCAM{our methodology} to show lower bounds for
problems parameterized by treewidth.  To this end, we make use of a
stricter notion of fpt-reductions, which \emph{linearly preserves} the
parameter.
\FIX{%
  Given functions~$f,g: \Nat \rightarrow \Nat$, where~$g$ is linear,
  and an
  $f$-bounded fpt-reduction $r$ using~$g$. Then, we call~$r$ an
  $f$-bounded~\emph{fptl-reduction} using~$g$. %
}%
 Next, we discuss the methodology
for proving lower bounds of a problem~$\mathtt{P}$ for treewidth 
consisting of the following.
\begin{enumerate}[leftmargin=0pt, itemindent=10pt, itemsep=5pt]
\item {\bf{Graph Representation:}} Pick a graph
  representation $G(I)$ for a given instance~$I$ of %
  problem~$\mathtt{P}$. 
\item {\bf Quantifier Rank:} Fix a quantifier rank~$\ell$ such that
  there is a function~$f: \Nat \rightarrow \Nat$ with
  $f(k) \in \bigO(\tower(\ell,$ $k))$ and %
  aim for
  establishing lower bound~$\tower(\ell, \Omega(k))\cdot \poly(\CCard{I})$.
\item {\bf Establish Reduction:} Establish an $f$\hy
  bounded fptl-reduct\-ion from an arbitrary QBF~$\Q$ of quantifier rank~$\ell$
  parameterized by treewidth of the \FIXCAM{primal graph} %
of~$\Q$ %
  to an instance~$I$ of~$\mathtt{P}$ parameterized by treewidth as well.
\item {\bf Conclude lower bound:} Then, by applying
  Theorem~\ref{lab:primqbflb} conclude that unless ETH fails, an~arbitrary
  instance~$I$ of problem~$\mathtt{P}$ cannot be solved in
  time~$\tower(\ell,$ $o(k))\cdot \BUGFIX{\poly(\CCard{I})}{2^{o(\CCard{I})}}$
  where~$k=\tw{G(I)}$.
\end{enumerate}
\noindent We can generalize this to \FIXCAM{``non-canonical''} lower bounds. To this end, one aims in Step 2 for a lower bound of the \FIXCAM{form $\tower\left(\ell, \Omega(g^{-1}(k)\right)\cdot \poly(\CCard{I})$} for some function~$g: \Nat \rightarrow \Nat$ such that~$g^{-1}$ is well-defined, and
\FIXCAM{$f(k)\in \bigO\left(\tower(\ell, g^{-1}(k))\right)$}.
Then, in Step 3 one needs to establish an \FIX{$f$-bounded} fpt-reduction
using~$g$ accordingly, in order to conclude in Step 4 that under ETH an
arbitrary instance~$I$ of~$\mathtt{P}$ cannot be solved in
time~\FIX{$\tower(\ell, o(g^{-1}(k)))\cdot \BUGFIX{\poly(\CCard{I})}{2^{o(\CCard{I})}}$},
where~$k=\tw{G(I)}$. %

With the help of this methodology one can show lower bounds~$f(k)$ for certain problems~$\mathtt{P}$, parameterized by tree\-width, by reducing from the canonical $\ell$\hy\QBFSAT problem
parameterized by treewidth~$k$ as well.
Thus, one avoids directly using ETH via tedious reductions from \SAT, which involves problem-tailored constructions of instances of~$\mathtt{P}$ whose treewidth is $\ell$-fold logarithmic in the number of variables or clauses of the given \SAT formula.

\FIXCAM{Note that the methodology naturally extends to pathwidth,
since the result of Theorem~\ref{lab:primqbflb} 
easily extends to ($2$-local) pathwidth by construction of our reduction~$R$, which works for \emph{any} tree decomposition including the special case of path decompositions.}
\FIXCAM{Formal details will be provided in Section~\ref{sec:correctness} on correctness in Corollary~\ref{cor:pw}, followed by further consequences of Theorem~\ref{lab:primqbflb}.}

\smallskip\medskip\noindent\textbf{Showcases.}
\noindent \FIX{Table~\ref{tab:metho} gives a brief overview of selected problems and their respective runtime lower bounds under ETH.}
Then, the proof of Theorem~\ref{thm:countingsucks} below serves as an example for applying the methodology, 
showing that Theorem~\ref{lab:primqbflb} also allows for quite general results on projection.
\FIX{Note that this bounds are tight under ETH.}
{%
	\begin{table}[t]%
		\fontsize{8}{9}\selectfont
		\centering%
		\begin{tabular}[t]{@{\hspace{0em}}\shortversion{p{17.25em}}\longversion{p{20em}}@{\hspace{.1em}}|@{\hspace{.1em}}p{2.2em}@{\hspace{.1em}}|@{\hspace{.05em}}p{3.8em}@{\hspace{.1em}}|@{\hspace{.1em}}p{2.22em}@{\hspace{.1em}}|@{\hspace{.1em}}p{1.5em}@{\hspace{.1em}}|@{\hspace{.1em}}p{1.5em}@{\hspace{.1em}}|@{\hspace{.1em}}p{1em}@{\hspace{0em}}}%
			\multirow{2}{*}{Problem~$\mathtt{P}$} &
			\multicolumn{6}{c}{LB \FIX{$\tower(i,\Omega(k)) \cdot \poly(\CCard{I})$}}\\
			& {\centering$i{=}1$} & \centering$i{=}2$ & \centering$i{=}3$ & \centering$i{=}4$ & \centering$i{=}5$ & $i{=}\ell$ \\\hline
			\problemFont{Min Vertex Cover / Dominating Set~\cite{GareyJohnson79}} & $\triangledown$\cite{ImpagliazzoPaturiZane01} &&&&\\
			\problemFont{Max Indep.\ Set, Hamilt. Cycle~\cite{GareyJohnson79}} & $\triangledown$\cite{ImpagliazzoPaturiZane01}  &&&&\\
			\problemFont{3-Colorability, \SAT, \sharpSAT~\cite{GareyJohnson79,SamerSzeider10b}} & $\triangledown$\cite{ImpagliazzoPaturiZane01} &&&&\\
			\problemFont{$\Circ$, $\PAP$~\cite{LampisMitsouMengel18}} & & $\triangledown_t$\cite{LampisMitsouMengel18} $\blacktriangledown$ &&&\\
			\problemFont{MUS~\cite{LampisMitsouMengel18}} & & $\triangledown_t$\cite{LampisMitsouMengel18} $\blacktriangledown$ &&&\\
			\problemFont{$\Skep_{\preferred}, \Skep_{\semistable}, \Cred_{\semistable}$~\cite{FichteHecherMeier18c}} & & $\triangledown_t$\cite{FichteHecherMeier18c} $\blacktriangledown$ &&&\\
			\problemFont{\PMC~\cite{FichteEtAl18}} & & $\triangledown_t$\cite{FichteEtAl18} $\blacktriangledown$ &&&\\
			\problemFont{\cASP, \sharpASP~\cite{JaklPichlerWoltran09,FichteHecher18c}} & & $\triangledown_t$\cite{FichteHecher18c} $\blacktriangledown$ &&&\\
			\problemFont{$k$-Choosability~\cite{MarxMitsou16}, $k\geq 3$} & & $\triangledown$\cite{MarxMitsou16} $\blacktriangledown$ &&&\\
			\problemFont{$k$-Choosability Deletion~\cite{MarxMitsou16}, $k\geq 4$} & & & $\triangledown$\cite{MarxMitsou16} &&\\
			\problemFont{\PPAP} & & & $\blacktriangledown$ &&\\
			\problemFont{\PASP~\cite{Aziz15a,FichteHecher18c}} & & & $\blacktriangledown$\cite{FichteHecher18c} &&\\
			\problemFont{$\PCC_\SEM$~\cite{FichteHecherMeier18c}, $\SEM \in \{\preferred, \semistable, \stage\}$} & & & $\blacktriangledown$\cite{FichteHecherMeier18c} &&\\
			\problemFont{Candidate World View Check~\cite{EiterShen17}} & & & $\blacktriangledown$ &&\\
			\problemFont{World View Check~\cite{EiterShen17}} & & & & $\blacktriangledown$ &\\
			\problemFont{\#Projected Guesses to World Views} & & & & & $\blacktriangledown$\\
			\problemFont{$\ell\hy\QBFSAT$, $\text{\#}\ell\hy\QBFSAT, \ell \geq 1$~\cite{Chen04a}} & & & & & & $\blacktriangledown$\\
			\problemFont{\PQBF~\cite{DurandHermannKolaitis05}: $\text{\#}\Sigma_{\ell - 1}\SAT, \text{\#}\Pi_{\ell - 1}\SAT, \ell\geq 2$} & & & & & & $\blacktriangledown$\\
		\end{tabular}%
		\caption{\FIX{Runtime lower bounds (under ETH) for selected problems, where $I$ denotes an instance of problem~$\mathtt{P}$ and~$k$ refers to the treewidth (``$\triangledown_t$''), pathwidth (``$\triangledown$''), or $2$-local pathwidth (``$\blacktriangledown$'') %
		 of the corresponding (primal) graph of~$I$.
		Results known from the literature are marked by ``$\triangledown_t$'' and ``$\triangledown$''.
		By ``$\blacktriangledown$'', we indicate that the result holds due to lower bound advancements and the methodology described in this paper. 
		We obtain results for ``$\blacktriangledown$'', with known lower bound (``$\triangledown_t$'', ``$\triangledown$''), by the existing lower bound proof together with our methodology~for~$2$-local pathwidth.
		Bounds are asymptotically tight unter ETH; 
		corresponding upper bounds (e.g.,\cite{SamerSzeider10b,JaklPichlerWoltran09,LampisMitsouMengel18,MarxMitsou16,Chen04a,Bodlaender88,FichteHecherMeier18c,FichteHecher18c,HecherMorakWoltran20}) are out of scope.  
		\FIXCAM{For definitions, we refer to the %
		problem compendium in an online self-archived version.} %
		}
		}
		\shortversion{\vspace{-3em}}
		\label{tab:metho}
	\end{table}%
}

\begin{THM}\label{thm:countingsucks}
Given an open QBF of the form~$\Q=Q_1 V_1. {Q_2} V_2.$ ${Q_3} V_3 \cdots Q_\ell V_\ell. F$
where~$\ell\geq 1$, and $F$ is a 3-CNF formula (if~$Q_\ell=\exists$), or a 3-DNF formula (if~$Q_\ell=\forall$).
Then, under ETH, \PQBF is indeed harder than deciding validity of~$\Q[\iota]$ for any assignment~$\iota: \fvar(\Q) \rightarrow \{0,1\}$.
In particular, assuming ETH, %
\PQBF cannot be solved in time~$\tower(\ell+1,o(k))\cdot \BUGFIX{\poly(\Card{\var(F)})}{2^{o(\Card{\var(F)})}}$, where~$k$ is the pathwidth of the primal graph~$P_\Q$.
\end{THM}
\begin{proof}%
Assume towards a contradiction that under ETH one can solve projected model counting of~$\Q$ in time~$\tower(\ell+1,o(k))\cdot \BUGFIX{\poly(\Card{\var(F)})}{2^{o(\Card{\var(F)})}}$.
\FIXCAM{In the following, we define \FIXCAM{an fptl-reduc\-tion~$r$} from~$\QBFSAT$ to the decision variant~$\PQBF$-at-least-$u$ of~$\PQBF$, where a given open QBF~$\Q$ is a yes instance if and only if the solution (count) to~$\PQBF$ of~$\Q$ is at least~$u$.}

\FIXCAM{In particular, we transform a closed QBF~$\Q'={{Q_0}} V_0. Q_1 V_1.$ ${Q_2} V_2. {Q_3} V_3 \cdots Q_\ell V_\ell. F$, where~$k$ is the pathwidth of~$P_\Q$ to an instance~$\Q=Q_1. V_1. {Q_2} V_2. {Q_3} V_3 \cdots Q_\ell V_\ell. F$ of~$\PQBF$-at-least-$u$, where~$\fvar(\Q) = V_0$, and we set~$u\eqdef 1$ if ${Q_0}=\exists$ and~$u\eqdef 2^{\Card{V_0}}$, otherwise.
The reduction is indeed correct, since~$\Q'$ is a yes-instance of~$\QBFSAT$ if and only if~$\Q=r(\Q')$ is a yes-instance of~$\PQBF$-at-least-$u$.}
Then, one can solve~$\Q'$ of quantifier rank~$\ell+1$ in time~$\tower(\ell+1,o(k))\cdot \BUGFIX{\poly(\Card{\var(F)})}{2^{o(\Card{\var(F)})}}$, which contradicts Theorem~\ref{lab:primqbflb} and Corollary~\ref{cor:pw}.
\end{proof}

\begin{COR}
Assuming ETH, an instance~$Q$ of the problem~$\text{\#}\Sigma_\ell\SAT$ or~$\text{\#}\Pi_\ell\SAT$ cannot be solved in time~$\tower(\ell+1,o(k))\cdot \BUGFIX{\poly(\Card{\var(\matr(\Q))})}{2^{o(\Card{\var(\matr(Q))})}}$, where~$k$ is the pathwidth of~$P_Q$.
\end{COR}

\noindent Next, we provide further examples (listed in Table~\ref{tab:metho}) of the applicability of our methodology.
\FIXCAM{We provide brief definitions of the problems discussed below in an online self-archived version.} %

\begin{PROP}[cf., \cite{FichteHecher18c}]
Unless \FIXCAM{ETH fails, \PASP} for given ASP program~$\Pi$ and a set~$P\subseteq\var(\prog)$ of projection variables cannot be solved in time~$\tower(3,o(k))\cdot \BUGFIX{\poly(\Card{\Pi})}{2^{o(\CCard{\Pi})}}$, where~$k$ is the pathwidth of the primal graph\footnote{For a definition of the primal graph of a program, we refer to~\cite{JaklPichlerWoltran09}.} of~$\Pi$.
\end{PROP}
\begin{proof}[Proof (Idea).]
Fptl-reduction from~$\forall\exists\forall$\hy\SAT to \PASP, both parameterized by the pathwidth of \FIXCAM{its primal graph}.
\end{proof}

\begin{PROP}[cf., \cite{FichteHecherMeier18c}]
Let $\SEM \in \{\preferred, \semistable, \stage\}$ and~$F$ be an argumentation framework. Unless
  \FIXCAM{ETH fails, we} cannot solve the problem~$\PCC_{\SEM}$ in
  time~$\tower(3,o(k)) \cdot \BUGFIX{\poly(\CCard{F})}{2^{o(\CCard{F})}}$ where~$k$ is the
  pathwidth of~$F$ (underlying graph).
\end{PROP}
\begin{proof}[Proof (Idea).]
Fptl-reduction from~$\forall\exists\forall$\hy\SAT parameterized by pathwidth of primal graph, to $\PCC_{SEM}$ (parameterized by pathwidth of the underlying graph).
\end{proof}

\begin{THM}
Unless ETH fails, \PPAP for given instance $(T,H,M)$ and set~$P$ of projection variables cannot
be solved in time~$\tower(3,o(k))\cdot \BUGFIX{\poly(\Card{\var(T)})}{2^{o(\Card{\var(T)})}}$, where~$k$ is the pathwidth of the primal graph of~$T$.
\end{THM}
\begin{proof}[Proof (Idea).]
Fptl-reduction from~$\forall\exists\forall$\hy\SAT to the decision variant of \PPAP, 
either from scratch or by lifting
 existing reduction~\cite{LampisMitsouMengel18} from~$\exists\forall$\hy\SAT.
\end{proof}

\begin{THM}
Given an epistemic program~$\prog$ and a variable~$a\in\var(\prog)$. 
Then, unless ETH fails, deciding the problem \problemFont{Candidate World View Check}
cannot be solved in time $\tower(3,o(k))\cdot \BUGFIX{\poly(\Card{\prog})}{2^{o(\Card{\prog})}}$, and
the problem \problemFont{World View Check} for~$a$ cannot be solved in time~$\tower(4,o(k))\cdot \BUGFIX{\poly(\Card{\prog})}{2^{o(\Card{\prog})}}$, where $k$ is the pathwidth of the primal graph of~$\prog$.
\end{THM}
\begin{proof}[Proof (Idea).]
Fptl-reduction from~$\exists\forall\exists$\hy\SAT, or~$\exists\forall\exists\forall$\hy\SAT (parameterized by pathwidth of primal graph), respectively.
Actually the reductions from the literature for showing~$\Sigma_3^P$-hardness and~$\Sigma_4^P$-hardness~\cite{EiterShen17} form fptl-reductions.
\end{proof}

\noindent \FIX{Note that our work focuses on lower bounds.
However, \FIXCAM{the corresponding upper bounds} for treewidth can be established by reductions to $\ell$\hy\QBFSAT to obtain asymptotically tight results under ETH, see Table \ref{tab:metho}.}

\section{Correctness, Compression and Runtime}\label{sec:correctness}

In the following, we show correctness and properties of our reduction presented in Section~\ref{sec:reduction}. Therefore, we assume a given QBF~$\Q\eqdef  Q_1 V_1. Q_2 V_2.\cdots \forall V_\ell. D$, where~$D$ is in~3-DNF.
Further, let~$\mathcal{T}=(T,\chi,\delta)$ such that~$T=(N,\cdot,r)$ be a \FIXCAM{labeled almost $c$-nice tree decomposition of primal graph~$P_\Q$ of width~$k$}.
The reduced instance is addressed by~$R(\Q)$, where reduction~$R$ is defined as in Section~\ref{sec:reduction}.
The resulting \FIX{QBF of quantifier rank $\ell+1$} is referred to by~$R'(R(\Q))$ and its matrix in 3-CNF is given by~$C=\matr(R'(R(\Q)))$.

To simplify presentation, we introduce the following definitions.
Let~$d\in D$ be a term, $t\in N$ be a node of the tree decomposition,
and~$1\leq j\leq 3$, then
$\bterm{d}{t}\eqdef \bigcup_{1\leq j\leq 3}[\bvali{\ato{d}{j}}{t}{j}
\cup \{\bvv{d}{t}{j}\}]$.  Further, given an
assignment~$\alpha: \var(D) \rightarrow \{0,1\}$, we define a
function~$\local{\cdot}$ that produces new assignments to copies of the
variables, therefore let
$\local{\alpha} \eqdef \{x_t \mapsto \alpha(x)\mid x\in \dom(\alpha),
t\in\nint{\mathcal{T}}{x}\}$ denote the \emph{matching} assignment of
the corresponding guess variables.
\FIXCAM{
  Further, for a set~$\mathcal{S}$ of literals and an assignment~$\iota$,
we say assignment~$\iota$ \emph{respects~$\mathcal{S}$}, if~$(\bigwedge_{l\in\mathcal{S}}l)[\iota]$ evaluates to true. %
}

\smallskip\noindent\textbf{Correctness.}
\FIXCAM{Next}, we establish correctness of reduction~$R$.
\begin{LEM}\label{lem:slavesandmasterequiv}
  \FIXCAM{Let $\kappa$ be any
  assignment of at least two variables $x_t, x_{t'}\in\lint{\mathcal{T}}{\var(D)}$
  such that $\kappa(x_t)\neq \kappa(x_{t'})$ for nodes $t, t' \in \nint{\TTT}{x}$
  with~$x\in\var(D)$. %
  Then, $R(\Q)[\kappa]$ is invalid.}
\end{LEM}
\begin{proof}
  \FIX{We construct an assignment
  $\kappa'$ %
  which extends~$\kappa$ and sets certain
  variables in~$\lbv{\mathcal{T}}$. Then, we show that %
  $R(\Q)[\kappa']$ is invalid, which suffices 
  since~$\lbv{\mathcal{T}}$ is universally quantified.}
  The construction of~$\kappa'$ is as follows: for every~$b\in \bval{x}{t''}$
  and every~$t''\in N$ where~$x\in\chi(t'')$, we set~$\kappa'(b)\eqdef 1$,
  if~$b$ is a variable; and~$\kappa'(b)\eqdef 0$, otherwise.
  Assume towards contradiction that there is an
  assignment~$\kappa'': \lint{\mathcal{T}}{\var(D)} \cup
  \lbv{\mathcal{T}} \cup \lbvv{\mathcal{T}}\rightarrow \{0,1\}$ that
  extends~$\kappa'$ such that~$R(\Q)[\kappa'']$ is valid.  In
  particular, the assignment~$\kappa''$ sets variables
  in~$\lbvv{\mathcal{T}}$ such that every clause \FIXCAM{in the assigned variables of}~$R(\Q)$, in
  particular, parts $\mathcal{G}$ and $\mathcal{P}$, is valid under
  the assignment~$\kappa''$.
  By Condition~(ii) of the definition of a tree decomposition
  (connectedness), $\mathcal{T}[x]$ induces a connected tree as well.
  In consequence, irrelevant of how $\kappa''$ assigns
  variable~$v_{r'}$ for the root node~$r'$ of~$\mathcal{T}[x]$, the
  clauses in Formulas~(\ref{red:propatom})
  \FIXCAM{and~(\ref{red:negpropatom}) enforce that exactly the same truth
  value~$v=\kappa''(v_{n})$ has to be set for any
  node~$n\in\nint{\mathcal{T}}{x}$.
  Then, $\kappa''(v_{t})=v$ and~$\kappa''(v_{t'})=v$ holds. By
  part~$\mathcal{G}$ of our reduction, more precisely,
  Formulas~(\ref{red:guessatom}) and~(\ref{red:guessnegatom}), we
  conclude that both~$\kappa''(x_t)=v$ and~$\kappa''(x_{t'})=v$, which
  contradicts that~$\kappa(x_t) \neq \kappa(x_{t'})$.}
\end{proof}
\begin{LEM}\label{lem:satbyfindingterm}
Given an assignment~$\iota: \lint{\mathcal{T}}{\var(D)} \cup \lbv{} \cup \lbvv{} \rightarrow \{0,1\}$.
\FIXCAM{Then, for any assignment~$\kappa: \lint{\mathcal{T}}{\var(D)} \cup \lbv{} \cup \lbvv{} \cup \lbvd{} \cup \lbvvd{} \rightarrow \{0,1\}$ that extends~$\iota$,
$R(\Q)[\kappa]$ is invalid, if (a) there is no term~$d_i\in D$ with~$t=\lnode{d_i}$ such that~$\kappa$ respects~$\bterm{d_i}{t}$.}
Now assume that there is~$d_i\in D$ \FIXCAM{with~$\kappa$} respecting~$\bterm{d_i}{\lnode{d_i}}$, then $R(\Q)[\kappa]$ is also invalid, if (b) $\kappa(v_{t,j}) \neq \kappa(x_{t'})$, where $x=\ato{d_i}{j}$ for \FIXCAM{some}~$1\leq j \leq 3$ and~$t'\in\nint{\mathcal{T}}{x}$.
\end{LEM}
\begin{proof}
Assume towards a contradiction that (a) is not the case, i.e., there is no~$d_i\in D$ such that~$\kappa$ respects~$\bterm{d_i}{t}$ for~$t=\lnode{d_i}$ and still~$R(\Q)[\kappa]$ is valid. Observe that by~$R(\Q)$, in particular, by construction of the check part~$\mathcal{CK}$ of~$R$, $\kappa(sat_{\leq r})=1$ by~(\ref{chk_satn}) and therefore~$\kappa(sat_t)=1$ by~(\ref{chk_leqtright}) for at least one node~$t\in N$ has to be set in~$\kappa$. 
This, however, implies by~(\ref{chk_negsatt}) that~$t=\lnode{d_i}$ for some~$d_i\in D$.
In consequence, by construction of~(\ref{chk_satcb}) and~(\ref{chk_satcbv}), $\kappa$ respects~$\bterm{d_i}{t}$, where~$t{=}$ $\lnode{d_i}$, contradicting the assumption.

Towards contradicting~(b), assume that there is~$d_i\in D$ with~$t=\lnode{d_i}$ and~$x=\ato{d_i}{j}$ as well as~$t'\in\nint{\mathcal{T}}{x}$ such that~$\kappa(v_{t,j}) \neq \kappa(x_{t'})$ and still~$R(\Q)[\kappa]$ is valid. 
Observe that for any two nodes~$t'',t'''\in\mathcal{T}[x]$, $\kappa$ respects~$\bvali{x}{t''}{j}$ and~$\bvali{x}{t'''}{j}$ by~(\ref{red:propatomj}) and connectedness of~$\mathcal{T}[x]$.
Further, for any $t'',t'''\in\mathcal{T}[x]$, \FIXCAM{$\kappa(v_{t'',j})=\kappa(v_{t''',j})$} by~(\ref{red:propvright}).
Then, since~(\ref{red:guessatomj}) and~(\ref{red:guessnegatomj}) ensure that~\FIXCAM{$\kappa(x_{t'}) = \kappa(v_{t',j})$},
ultimately by connectedness of~$\mathcal{T}[x]$,%
~$\kappa(v_{t,j})=\kappa(x_{t'})$ holds. %
\end{proof}
\begin{THM}[Correctness]%
\label{lem:correctness}
  Let~$\Q$ be a QBF of the
  form~$Q= Q_1 V_1. {{Q_2}} V_2. \cdots \forall V_\ell. D$ where $D$
  is in DNF. Then, \FIXCAM{for any 
  assignment~$\alpha: {\fvar(\Q)} \rightarrow \{0,1\}$, we have $\Q[\alpha]$ is valid if and only if
  $R(\Q)[\alpha']$ is valid, where assignment: %
$\fvar(R(\Q)) \rightarrow \{0,1\}$ is such that
$\alpha' = \local{\alpha}$}.
\end{THM}
\begin{proof}
Let $\TTT=(T, \chi,\delta)$ be the labeled tree decomposition that is
  computed when constructing~$R$,  where~$T=(N,A,r)$.
  We proceed by induction on the quantifier rank~$\ell$.

  \noindent \emph{Base case}. Assume~$\ell=1$.\nopagebreak

  ``$\Longrightarrow$'': Let $\alpha$ be an assignment to the free
  variables of~$\Q$ for which $\Q[\alpha]$ is valid. Further, let
  $\alpha' \eqdef \local{\alpha}$. We show that $R(\Q)[\alpha']$ is
  valid as well. Let therefore~$\iota$ be an arbitrarily chosen
  assignment to the variables in $V_1$. Since $\ell = 1$, we have
  $Q_1 = \forall$.
  We define an
  assignment~$\kappa: \lint{\mathcal{T}}{\forall, V_1, \emptyset}
  \rightarrow \{0,1\}$ such that $\kappa(x_t) \eqdef \iota(x)$ for
  every~$x_t\in \lint{\mathcal{T}}{\forall, V_1, \emptyset}$ with
  $t \in N$ and $x\in \var(D)$.
  Next, we define an
  assignment~$\kappa': \lint{\mathcal{T}}{\forall, V_1, \emptyset}
  \cup \lint{\mathcal{T}}{\exists,$ $\emptyset, V_1} \rightarrow
  \{0,1\}$ that extends~$\kappa$ and
  sets~$\kappa'(x_{t'}) \eqdef \iota(x)$ for every
  $x_{t'}\in\lint{\mathcal{T}}{\exists, \emptyset, V_1}$ with
  $t' \in N$.
  \FIXCAM{Assignment $\kappa'$ has by construction the same
  truth value for each of the copies~$x_t$ of~$x$, which is needed %
  for $R(\Q)[\alpha'\cup\kappa]$ to be valid
in order to not contradict Lemma~\ref{lem:slavesandmasterequiv}.}

  Then, we construct an assignment~$\kappa''$, which extends~$\kappa'$
  by the variables in~$\lbvd{\mathcal{T}}$, $\lbvvd{\mathcal{T}}$, and
  $\lsat$.
  By construction of~$\iota$ and since $\Q[\alpha]$ is valid,
  $Q[\alpha\cup \iota]$ is valid, which is the same as
  $D[\alpha\cup \iota]$ is valid. In consequence, as $D$ is in DNF,
  there is at least one term~$d \in D$ such
  that~$d [\alpha\cup \iota]$ is valid.
  Depending on the term~$d$, we assign the variables
  in~$\lbvd{\mathcal{T}}$, $\lbvvd{\mathcal{T}}$, and $\lsat$ with
  assignment~$\kappa''$.
  By Definition~\ref{def:labeledTD}, there is a unique
  node~$t=\lnode{d}$ in the labeled tree decomposition for the
  term~$d$.
  Then, we
  set~$\kappa''(sat_{\leq t}) \eqdef \kappa''(sat_t) \eqdef 1$. For
  every ancestor~$t'$ of $t \in N$, we
  assign~$\kappa''(sat_{\leq t'})\eqdef 1$.
  For every node~$s \in N$ that is not an ancestor of~$t$, we
  set~$\kappa''(sat_{\leq s})\eqdef 0$. Finally, for every node~$u$,
  where~$u \neq t$, we set $\kappa''(sat_{u})\eqdef 0$.
  For every node~$t \in N$ and~$1\leq j\leq 3$
  with~$\ato{d}{j}\not\in\chi(t)$, we set~$\kappa''$ such that it
  respects~$\bvali{nil}{t}{j}\cup \{\bvv{d}{t}{j}\}$.  Finally, for
  every node~$t$ and~$1\leq j\leq 3$ with~$\ato{d}{j}\in\chi(t)$, we
  set~$\kappa''$ such that it \FIXCAM{respects~$\bterm{d}{t}$.}

  It remains to prove that for every
  assignment~$\beta: \lbv{\mathcal{T}} \rightarrow \{0,1\}$, there is
  an assignment~$\zeta: \lbvv{\mathcal{T}} \rightarrow \{0,1\}$ for
  which $R(\Q)[\alpha'\cup \kappa''\cup \beta \cup \zeta]$ is valid.
  \FIXCAM{For every variable~$x\in\chi(t)$, if for every node~$t \in N$,}
  assignment~$\beta$ respects $\bval{x}{t}$, then we set
  $\zeta(v_t)\eqdef (\alpha \cup \iota)(x)$. Otherwise,
  $\zeta(v_t) \eqdef 0$, since we can assign any truth value here.
  By construction of~$\alpha'$ and~$\kappa''$, clauses in
  Formulas~(\ref{red:guessatomj}) and~(\ref{red:guessnegatomj}) are
  satisfied of~$\mathcal{G}$.
  Every clause of~$\mathcal{CK}$ of~$R(\Q)$ is satisfied by
  construction of~$\kappa''\setminus \kappa'$ (and also by $\kappa''$)
  and $\zeta$.
  Clauses in Formulas~(\ref{red:propatomj})
  and~(\ref{red:propvright}) of~$\mathcal{P}$ are satisfied
  by~$\kappa''\setminus \kappa'$.
  Further, clauses in Formulas~(\ref{red:guessatom})
  and~(\ref{red:guessnegatom}) of~$\mathcal{G}$ are satisfied because
  of $\beta$, $\kappa''$, $\alpha'$, and $\zeta$.  
  Finally, the clauses in Formulas~(\ref{red:propatom})
  and~(\ref{red:negpropatom}) of~$\mathcal{P}$ are satisfied by
  construction of~$\beta$, $\zeta$, and $\kappa'$.

  ``$\Longleftarrow$'':
  Let $\alpha$ be an assignment to the free variables of $Q$ for which
  $Q[\alpha]$ is invalid.  We show that if QBF~$\Q[\alpha]$ is
  invalid, then~$R(\Q)[\alpha']$ is invalid as
  well. %
  \FIXCAM{Since~$\Q[\alpha]$} is invalid, $\Q[\alpha\cup \iota]$ is invalid for
  \FIXCAM{any} assignment~$\iota: V_1 \rightarrow \{0,1\}$.  Assume towards a
  contradiction that~$R(\Q)[\alpha']$ is valid.
  We define an assignment~$\kappa\eqdef \local{\iota}$,
  which is
  $\kappa: \lint{\mathcal{T}}{\forall, V_1, \emptyset} \cup
  \lint{\mathcal{T}}{\exists, \emptyset, V_1} \rightarrow \{0,1\}$
  such that $\kappa(x_t) \eqdef \iota(x)$ for
  every~$x_t\in \lint{\mathcal{T}}{\forall, V_1, \emptyset} \cup
  \lint{\mathcal{T}}{\exists, \emptyset, V_1}$ with $t \in N$ and
  $x\in \var(D)$.
  \FIXCAM{Observe that %
  by Lemma~\ref{lem:slavesandmasterequiv},
  $\kappa$ is the only remaining option to obtain valid~$R(\Q)[\alpha]$.
  As a result, since~$R(\Q)[\alpha']$ is claimed valid,} $R(\Q)[\alpha'\cup \kappa]$ is valid as well.
  In consequence, by Lemma~\ref{lem:satbyfindingterm} Statement~(a),
  there \FIXCAM{has to exist} an extension~$\kappa'$ of~$\alpha'\cup \kappa$ such that
  for some~$d\in D$, $\kappa'$ respects~$\bterm{d}{t}$,
  where~$t =\lnode{d}$.
  By Lemma~\ref{lem:satbyfindingterm} Statement~(b), for~$1\leq j\leq 3$ and every
  node~$t'\in \nint{\TTT}{y}$, where $y\eqdef \ato{d}{j}$, we have
  $\kappa'(v_{t,j}) = \kappa'(y_{t'})$.
  \FIXCAM{However, }%
  by construction of~$\kappa'$ and connectedness
  of~$\mathcal{T}[y]$, \FIXCAM{then}~$(\alpha\cup \iota)$ respects
  $d$. In consequence, this contradicts our assumption
  that~$\Q[\alpha\cup\iota]$ is invalid.
  \noindent \emph{Induction step} ($\ell > 1$): We assume that the
  theorem holds for a given~$\ell-1$ and it remains to prove that it
  then holds for~$\ell$.

  ``$\Longrightarrow$'': We proceed by case distinction on the first
  quantifier, i.e., (Case~1)~$Q_1 = \exists$ and
  (Case~2)~$Q_1 = \forall$. Thereby, we show that if~$\Q[\alpha]$ is
  valid and has quantifier rank~$\ell$, then~$R(\Q)[\alpha']$ is
  valid as well.

  (Case~1) $Q_1=\exists$: Since~$\Q[\alpha]$ is valid, we can
  construct at least one assignment~$\iota: V_1 \rightarrow \{0,1\}$
  such that~$\Q[\alpha\cup \iota]$ is valid.  By induction hypothesis,
  since the QBF~$\Q[\alpha\cup \iota]$ has quantifier rank~$\ell-1$,
  and is valid, there are~$\alpha',\iota'$ such
  that~$R(\Q)[\alpha'\cup \iota']$ is valid as well.  In particular,
  by induction hypothesis,
  $\alpha'= \local{\alpha}, \iota'= \local{\iota}$ and
  therefore~$R(\Q)[\alpha']$ is valid as well.

  (Case~2) $Q_1=\forall$: Since~$\Q[\alpha]$ is valid, for \FIXCAM{any}
  assignment~$\iota$ of~$V_1$, we obtain that~$\Q[\alpha\cup \iota]$
  is valid. \FIXCAM{In the following, we denote by~$R''(\Q)$ the QBF
that is obtained from~$R(\Q)$, where variables in~$\lint{\mathcal{T}}{\exists, \emptyset, V_1}$
do not appear in the scope of a quantifier, i.e., these variables,
while existentially quantified in~$R(\Q)$,
are free variables in~$R''(\Q)$. }   By induction hypothesis, since the
  QBF~$\Q[\alpha\cup \iota]$ has quantifier rank~$\ell-1$, and is
  valid, there are~$\alpha',\iota'$ \FIXCAM{with $\alpha'= \local{\alpha}, \iota'= \local{\iota}$,}  such
  that~$R''(\Q)[\alpha'\cup \iota']$ is valid as well.
  Then,
  since $\iota$ was chosen arbitrarily, for every assignment~$\kappa$
  of variables
  in~$\dom(\iota')\cap \lint{\mathcal{T}}{Q_1, V_1, \emptyset}$, there
  \FIXCAM{is (by Lemma~\ref{lem:slavesandmasterequiv}, since~$R''(\Q)[\alpha'\cup\iota']$ is valid)
  an} assignment~$\kappa'$ of variables
  in~$\dom(\iota')\cap\lint{\mathcal{T}}{\exists, \emptyset, V_1}$
  such that~$R(\Q)[\alpha' \cup \kappa \cup \kappa']$ is valid.  In
  consequence,~$R(\Q)[\alpha']$ is valid as well.

  ``$\Longleftarrow$'': Again, we proceed by case distinction in order
  to show that if~$\Q[\alpha]$ is invalid and has quantifier
  rank~$\ell$,~$R(\Q)[\alpha']$ is invalid as well.

  (Case 1) $Q_1=\exists$: Since~$\Q[\alpha]$ is invalid, for every
  assignment~$\iota$ of variables~$V_1$ we have
  that~$\Q[\alpha\cup \iota]$ is also invalid.  By induction
  hypothesis, since the QBF~$\Q[\alpha\cup \iota]$ has quantifier
  rank~$\ell-1$, and is invalid, there are assignments~$\alpha'$ and
  \FIXCAM{$\iota'$ such that~$R(\Q)[\alpha'\cup \iota']$ is invalid as well, where
  $\alpha'= \local{\alpha}, \iota'= \local{\iota}$.
  Therefore~$R(\Q)[\alpha']$ is invalid,} \FIXCAM{since~$\iota$ was
  chosen arbitrarily and by Lemma~\ref{lem:slavesandmasterequiv} 
  $\iota'$ covers all relevant cases, where~$R(\Q)[\alpha']$ could be valid.} %

  (Case 2) $Q_1=\forall$: Since~$\Q[\alpha]$ is invalid, there is at
  least one assignment~$\iota$ of variables~$V_1$, such
  that~$\Q[\alpha\cup \iota]$ is invalid.  By induction hypothesis,
  since QBF~$\Q[\alpha\cup \iota]$ has quantifier rank~$\ell-1$, and
  is invalid, there are assignments~$\alpha'$ and $\iota'$ \FIXCAM{with
$\alpha'= \local{\alpha}, \iota'= \local{\iota}$,}
 such
  that~\FIXCAM{$R''(\Q)[\alpha'\cup \iota']$ is invalid ($R''$ defined above)}, either.  %
   By
  Lemma~\ref{lem:slavesandmasterequiv}, even for an
  assignment~$\iota''$ that restricts $\iota'$ to variables
  in~$\dom(\iota')\cap \lint{\mathcal{T}}{Q_1, V_1, \emptyset}$, there
  cannot be an assignment~$\kappa$ to variables
  in~$\dom(\iota')\cap\lint{\mathcal{T}}{\exists, \emptyset, V_1}$
  such that~$R(\Q)[\alpha' \cup \iota'' \cup \kappa]$ is valid.  In
  consequence,~$R(\Q)[\alpha']$ is invalid as well. %
\end{proof}
\noindent\textbf{Compression and Runtime.}
After having established the correctness of reduction $R$, we move on to showing that this reduction indeed compresses the treewidth of the resulting QBF $R(Q)$, as depicted in Figure~\ref{fig:transformation}. 
In particular, we prove this claim by constructing a tree decomposition $\mathcal{T}'$ of the primal graph of~$R(Q)$ and show its relation to \FIXCAM{labeled almost $c$-nice} tree decomposition~$\mathcal{T}$ of~$Q$, where~$\width(\mathcal{T})=\tw{P_Q}$. 
Then, we discuss runtime properties of the reduction.
\begin{LEM}[Compression]\label{lem:compr}
The reduction~$R$ exponentially decreases treewidth. 
In particular, $R'(R(\Q))$ constructs a QBF such that the treewidth of the primal graph of~$R'(R(\Q))$ is~$12\cdot\ceil{\log(k+1)}+7c+6$, where~$k$ is the treewidth of\FIXCAM{~$P_\Q$ and~$c\leq k$}.
\end{LEM}
\begin{proof}
Assume a labeled almost $c$-nice tree decomposition~$\mathcal{T}{=}$ 
$(T,\chi,r)$ of~$P_\Q$ of width~$k$, where~$T=(N,E)$.
From this we will construct a tree decomposition~$\mathcal{T}'=(T,\chi',r)$ of the primal graph of~$R(\Q)$.
For each tree decomposition node~$t\in N$ with~$\children(t)=\{t_1, \ldots, t_s\}$, we set its bag~$\chi'(t)\eqdef \{b\mid x\in\chi(t),  b\in \bval{x}{t} \cup \bval{x}{t_1} \cup \cdots \cup \bval{x}{t_s}\}  \cup\{b \mid x\in\chi(t),  1 \leq j\leq 3, b\in \bvali{x}{t}{j} \cup \bvali{x}{t_1}{j} \cup \cdots \cup \bvali{x}{t_s}{j}\} \cup  \{x_{t'} \mid x_{t'}\in \lint{\mathcal{T}}{V}, t'=t\} \cup \bigcup_{t'\in \{t, t_1, \ldots, t_s\}}\{v_{t',1}, v_{t',2}, v_{t',3},v_{t'},$ $sat_{t'}, sat_{\leq t'}\}$.
Observe that all the properties of tree decompositions are satisfied.
In particular, connectedness is not destroyed since the only elements that are shared among (at most two) different tree decompositions nodes are in~$\lbv{\mathcal{T}}$, $\lbvv{\mathcal{T}}$ and in~$\lbvd{\mathcal{T}}$, and~$\lbvvd{\mathcal{T}}$.

Each bag~$\chi'(t)$ contains bit-vectors~$\bval{x}{t}, \bval{x}{t_1}, \ldots, \bval{x}{t_s}$ for each~$x\in\chi(t)$, resulting in at most~$3\cdot\ceil{\log(k)}$ many elements, since each node can have at most~$s=2$ many children.
Further, each bag additionally consists of bit-vectors~$\bvali{x}{t}{j},$ $\bvali{x}{t_1}{j}, \ldots, \bvali{x}{t_s}{j}$ for each~$x\in\chi(t) \cup\{nil\}$,
where~$1\leq j\leq 3$, which are at most~$3\cdot3\cdot\ceil{\log(k+1)}$ many elements. 
In total everything sums up to at most~$12\cdot\ceil{\log(k+1)}+7c+6$ many elements per node, since~$\Card{\{x_{t'} \mid x_{t'}\in \lint{\mathcal{T}}{V}, t=t'\}}\leq c$,~and moreover~$|\bigcup_{t'\in \{t, t_1, \ldots, t_s\}}\{v_{t',1}, v_{t',2}, v_{t',3},v_{t'},$ $sat_{t'},$ $sat_{\leq t'}\}|$ $\leq 6\cdot (c+1)$ due to~$\mathcal{T}$ being labeled and almost $c$-nice.
Note that the treewidth of~$R'(R(\Q))$ only marginally increases,
since there are at most~$\bigO(k\cdot\ceil{\log(k)})$ many clauses in each bag of~$\chi'(t)$ for any node~$t\in N$, each of size at most~$\bigO(\ceil{\log(k)})$.
However, the fresh variables, that were introduced during the 3-CNF reduction only turn up in at most two new clauses (that is, they have degree two in the primal graph).
Further, the construction can be controlled in such a way, that each new clause consists of at most two fresh variables.
In consequence, one can easily modify~$\mathcal{T}'$, by adding at most~$\bigO(k\cdot\ceil{\log(k)}^2)$ many \FIX{intermediate nodes for each node~$t\in N$}, such that the width of~$\mathcal{T}'$ is at most~$12\cdot\ceil{\log(k+1)}+7c+6$.
\end{proof}

\begin{THM}[Runtime]\label{lem:reduction_linear}
Given a QBF~$\Q$, where~$D{=}\matr(\Q)$, $k$ is the treewidth of the primal graph of~$\Q$. Then, constructing~$R'(R(\Q))$ takes time~$\bigO(2^{k^4}\cdot \FIXCAM{\CCard{D}}\cdot c)$\FIXCAM{, where~$c\leq k$}.
\end{THM}
\begin{proof}
First, we construct~\cite{Bodlaender96} a tree decomposition of the primal graph of~$\Q$ of width~$k$ in time~$2^{\bigO(k^3)}\cdot \Card{\var(D)}$, consisting of at most~$\bigO(2^{k^3}\cdot\Card{\var(D)})$ many nodes.
Then, we compute a labeled almost $c$-nice tree decomposition in time~\FIXCAM{$\bigO(k^2\cdot2^{k^3}\cdot(\CCard{D})$}~\cite[Lemma 13.1.3]{Kloks94a} without increasing the width~$k$, resulting in decomposition~$\mathcal{T}=(T,\chi)$, where~$T=(N,E,r)$ of the primal graph of~$\Q$.
Note that thereby the number of nodes is at most~\FIXCAM{$\bigO(k\cdot 2^{k^3}\cdot \CCard{D})$}.
The reduction~$R(\Q)$ then uses at most~\FIXCAM{$\bigO(k\cdot 2^{k^3}\cdot \CCard{D}\cdot c)$} many variables in~$\lint{\mathcal{T}}{V}$ since in almost $c$-nice tree decompositions one node ``introduces'' at most~$c$ variables. 
The other sets of variables used in~$R$ are bounded by~\FIXCAM{$\bigO(\ceil{\log(k+1)}\cdot k^2\cdot 2^{k^3}\cdot \CCard{D})$}.
Overall, there are \FIXCAM{$\bigO(\ceil{\log(k+1)}\cdot k^2\cdot 2^{k^3}\cdot \CCard{D}\cdot c)$} many clauses~constructed by~$R(\Q)$. 
Hence, the claim follows, since~$R'(R(Q))$ runs in time~\FIXCAM{$\bigO(\ceil{\log(k+1)}^2\cdot k^2\cdot 2^{k^3}\cdot \CCard{D}\cdot c)\subseteq \bigO(2^{k^4}\cdot \CCard{D}\cdot c)$}.
\end{proof}

\smallskip\noindent\textbf{Proof of the main result.} We are in position to prove the main result of this work. To this end, we show that the lower bounds are closed under negation and restate Theorem~\ref{lab:primqbflb}.

\begin{LEM}\label{lem:inverse}
Assume a given closed QBF of the form~$\Q=Q_1 V_1. {Q_2} V_2. Q_3 V_3 \cdots Q_\ell V_\ell. F$, where~$\ell\geq 1$ and~$F$ is in CNF if~$Q_\ell=\exists$, and~$F$ is in DNF if~$Q_\ell=\forall$.
Under ETH, one cannot solve~$\Q$ in time~$\tower(\ell,o(k))\cdot \BUGFIX{\poly(\Card{\var(F)
})}{2^{o(\Card{\var(F)})}}$ if and only if one cannot solve the negation~$\neg{Q}$ in \FIXCAM{the same time}.
\end{LEM}
\begin{proof}
Assume towards a contradiction that $\neg{Q}$ can be solved in time~$\tower(\ell,o(k))\cdot \BUGFIX{\poly(\Card{\var(F)})}{2^{o(\Card{\var(F)})}}$ under ETH. 
But then, since inverting the result can be achieved in constant time, under ETH we can solve~$\Q$ in time~$\tower(\ell,o(k))\cdot \BUGFIX{\poly(\Card{\var(F)})}{2^{o(\Card{\var(F)})}}$. Hence, we arrive at a contradiction. %
\end{proof}

\begin{restatetheorem}[lab:primqbflb]
\begin{THM}[QBF lower bound]
Given \FIXCAM{an arbitrary QBF} of the form~$\Q=Q_1 V_1. {Q_2} V_2. Q_3 V_3 \longversion\allowbreak\cdots Q_\ell V_\ell. F$
where~$\ell\geq 1$, and $F$ is a 3-CNF formula (if~$Q_\ell=\exists$), or~$F$ is a 3-DNF formula (if~$Q_\ell=\forall$).
Then, unless ETH fails, $\Q$ 
cannot be solved in time~$\tower(\ell,o(k))\cdot \BUGFIX{\poly(\Card{\var(F)})}{2^{o(\Card{\var(F)})}}$, where~$k$ is the treewidth of the primal graph~$P_\Q$.
\end{THM}
\end{restatetheorem}
\begin{proof}
We assume that~$\Q$ is closed, i.e., for $\Q$ we have~$\fvar(\Q)=\emptyset$.
We show the theorem by induction on the quantifier rank~$\ell$. For the induction base, where~$\ell=1$, the result follows from the ETH in case of~$Q_\ell=\exists$ since~$k\leq \Card{\var(F)}$.
If~$Q_\ell=\forall$, by Lemma~\ref{lem:inverse}, the result follows. Note that for the case of~$\ell=2$, the result has already been shown~\cite{LampisMitsou17} as well.

For the induction step, we assume that the theorem holds for given~$\Q$ of \FIX{quantifier rank~$\ell\geq 1$}, where~$\Q_\ell=\forall$, the treewidth of primal graph~$P_\Q$ is~$k$, and $F$ is in 3-DNF.
We show that then the theorem also holds for quantifier rank~$\ell + 1$.
Towards a contradiction, we assume that in general we can solve any QBF~$\Q'$ of quantifier rank~$\ell+1$,
in time~$\tower(\ell+1,o(\tw{P_{\Q'}}))\cdot \BUGFIX{\poly(\Card{\var(C')})}{2^{o(\Card{\var(C')})}}$, where $C'=\matr(\Q')$. %
\FIXCAM{We compute a labeled almost $c$-nice TD~$\mathcal{T}$ of~$P_\Q$ of width~$k$,
where~$c$ is in~$\mathcal{O}(\log(k))$.}
We proceed by case distinction on the last quantifier~$\Q'_{\ell+1}$ of~$\Q'$.

(Case 1) $\Q'_{\ell + 1}=\exists$:
Let~$\Q'= R'(R(Q))$, $C'=\matr(\Q')$ be the matrix of~$\Q'$, and~$k'$ be the treewidth of the primal graph of~$C'$.
Observe that~$\Q'$ has quantifier rank~$\ell+1$ and is of the required form.
By Lemma~\ref{lem:compr}, \FIXCAM{$k'={12\cdot\ceil{\log(k+1)}+7c+6}$}. 
As a result, \FIXCAM{since~$R$ is an fpt-reduction} \FIXCAM{(including time for computing~$\mathcal{T}$)} according to \FIXCAM{Theorem}~\ref{lem:reduction_linear}, one can solve~$\Q'$ in \FIXCAM{time~$\tower(\ell+1,o(12\cdot\ceil{\log(k+1)}+7c+6))\cdot \BUGFIX{\poly(\Card{\var(C')})}{2^{o(\Card{\var(C')})}}=\tower(\ell+1,o({\log(k)}))\cdot \BUGFIX{\poly(\Card{\var(C')})}{2^{o(\Card{\var(C')})}}$.} Therefore, by Theorem~\ref{lem:correctness} we can solve~$\Q$ in time~$\tower(\ell,o(k))\cdot \BUGFIX{\poly(\Card{\var(F)})}{2^{o(\Card{\var(F)})}}$, which contradicts the induction hypothesis. %

(Case 2) $\Q'_{\ell + 1}=\forall$: By Lemma~\ref{lem:inverse} one can decide in time $\tower(\ell+1,o(k))\cdot \BUGFIX{\poly(\Card{\var(C')})}{2^{o(\Card{\var(C')})}}$ whether~$\Q'$ is valid if and only if we can decide in time~$\tower(\ell+1,o(k))\cdot \BUGFIX{\poly(\Card{\var(C')})}{2^{o(\Card{\var(C')})}}$ whether~$\neg\Q'$ is valid. 
Note that after bringing~$\neg\Q'$ into prenex normal form, the last quantifier is~$\exists$. Therefore, the remainder of this case is (Case 1). 
{Hence, we have established the second case and this concludes the proof.}
\end{proof}
\noindent\FIX{\textbf{Further Consequences.} We generalize Theorem~\ref{lab:primqbflb} to the incidence
  graph.
  The \emph{incidence graph} of a formula~$F$ in CNF or DNF is the
  bipartite graph, which has as vertices the variables and clauses
  (terms) of~$F$ and an edge~$vc$ between every variable~$v$ and
  clause (term)~$c$ whenever~$v$ occurs in~$c$
  in~$F$~\cite{SamerSzeider10b}. We obtain the following.
}
\begin{COR}\label{cor:sinc}
  \FIXCAM{Given an arbitrary QBF~$\Q$ of quantifier rank $\ell\geq 1$.}
Then, under ETH one cannot
  solve~$\Q$ in time~$\tower(\ell,o(k))\cdot \BUGFIX{\poly(\Card{\var(\matr(\Q))})}{2^{o(\Card{\var(F)})}}$,
  where~$k$ is the treewidth of the incidence graph of~$\matr(\Q)$.
\end{COR}
\begin{proof}
The claim follows from Theorem~\ref{lab:primqbflb},
since, in general, the treewidth~$k$ of the incidence graph of~$\Q$ is bounded~\cite{SamerSzeider10b,FichteSzeider15} by treewidth~$k'$ of~$P_\Q$, i.e., $k \leq k' + 1$.
\FIXCAM{As a result, if the weaker lower bound of this corollary did not hold, Theorem~\ref{lab:primqbflb} would be violated.}
\end{proof}

\begin{COR}[$2$-Local Pathwidth bound]\label{cor:pw}
  Given \FIXCAM{an arbitrary} QBF of the
  form~$\Q=Q_1 V_1. {Q_2} V_2. Q_3 V_3 \cdots Q_\ell V_\ell. F$, where
  $\ell\geq 1$ and $F$ in 3-CNF (if~$Q_\ell=\exists$) or 3-DNF
  (if~$Q_\ell=\forall$).  Then, unless ETH fails, $\Q$ cannot be
  solved in time~$\tower(\ell,o(k))\cdot \BUGFIX{\poly(\Card{\var(F)})}{2^{o(\Card{\var(F)})}}$,
  where~$k$ is the $2$-local pathwidth of \FIXCAM{graph~$P_\Q$}.
\end{COR}
\begin{proof}
  First, we show the claim for \emph{pathwidth} by induction, 
  which can be easily established for base case~$\ell=1$.
  For the case of~$\ell=2$, related work~\cite{LampisMitsou17} %
  holds only for treewidth. %
  However, the cases for~$\ell\geq 2$ %
  follow from the proof of
  Theorem~\ref{lab:primqbflb}, since every path decomposition is also
  a tree decomposition, and the proofs of lemmas and theorems 
  used intermediately only rely on an
  \emph{arbitrary} tree decomposition. 
  To be more concrete, the proof of Theorem~\ref{lab:primqbflb}
  relies on Lemma~\ref{lem:compr}, whose proof shows compression for any tree decomposition,
  which hence also works for any PD as well. 
  Similarly, Theorem~\ref{lem:reduction_linear} also
  holds for pathwidth, since a PD of fixed pathwidth can be computed~\cite{Bodlaender96} even in time~$\mathcal{O}(2^{k^2}\cdot\Card{\var(F)})$, 
  and since computation of labeled almost $c$-nice decompositions works anologously for PDs. 
  Further, the remainder of the proof holds for the thereby obtained PD, since the construction works for any TD.
  Finally, Lemma~\ref{lem:inverse} holds independently of the parameter.
  As a result, reductions~$R'$ and $R$ used by Theorem~\ref{lab:primqbflb} indeed are sufficient for PDs.

  The proof above can be lifted to the case of $2$-local pathwidth
  by observing that the $2$-local pathwidth of $P_{R'(R(\Q))}$ is also bounded
  by~$12\cdot\ceil{\log(k+1)}+7c+6$ (cf. Lemma~\ref{lem:compr}),
  since each variable of~$P_{R'(R(\Q))}$ occurs at most
  twice in the constructed (path) decomposition~$\mathcal{T}'$ of~$R'(R(\Q))$.
\end{proof}
\vspace{-.7em}
\begin{REM}\label{rem:qcsp}
We remark that reduction~$R$ can be \FIXCAM{generalized}
to finite, \emph{non-Boolean} domains (QCSP, e.g.,~\cite{FergusonOSullivan07}). %
For given variables~$V$ of a QCSP formula~$\Q$, 
the variables in~$\lint{\mathcal{T}}{V}$, $\lbvv{\mathcal{T}}$,
and~$\lbvvd{\mathcal{T}}$ have to be made non-Boolean,
whereas the other variables used in~$R$ stay Boolean.
Consequently, one obtains similar results as in related work~\cite{LampisMitsou17},
but for quantifier rank $\ell\geq 3$.
\FIXCAM{We conject that} under ETH, validity of QCSPs~$\Q$ over domain~$\mathcal{D}$ of quantifier rank~$\ell$,
where~$k=\pw{P_\Q}$,
cannot be decided in time~$\tower(\ell - 1,|\mathcal{D}|^{o(k)})\cdot \BUGFIX{\poly(\Card{\var(\Q)})}{|\mathcal{D}|^{o(\Card{\var(\Q)})}}$.
\end{REM}

\FIX{%
  Finally, we establish a corollary that improves a result from the
  literature. To this end, we
  denote for given positive number~$n$ by~$\log^*(n)$ the smallest
  value~$i$ such that~$\tower(i,1)\geq n$. 
  \FIX{A known result~\cite[Corollary 1]{AtseriasOliva14a} states $\Sigma_{\ell}^P$-hardness for instances~$\Q$ of $(4\cdot \log^*(\Card{\var(\matr(Q))}))$-\QBFSAT, wheras here we establish para-$\Sigma_{\ell}^P$-hardness for instances~$\Q'$ of $(\log^*(|\var($ $\matr(\Q'))|))$-\QBFSAT.
  This is possible by applying our established reduction~$R$,
  which is rather fine-grained since it only increases quantifier rank by one,
  and it works indeed for any QBF, and not just for a certain classes of QBFs in contrast to the known result.
  As a consequence, whenever a new class of $\ell$-QBFs with a certain treewidth or pathwidth guarantee was discovered, which is still $\Sigma_{\ell}^P$-hard, one \emph{immediately} obtains para-$\Sigma_{\ell}^P$-hardness by using reduction~$R$. 
  Then, one could potentially further improve quantifier alternations \FIXCAM{by applying %
reduction~$R$, which is (asymptotically) tight under ETH.}
  }
\begin{COR}\label{cor:twhardness}%
  Given any integer~$\ell\geq 1$. Then, deciding~\QBFSAT
  is~para-$\Sigma_{\ell}^P$-hard when parameterized by ($2$-local) pathwidth of
  the primal graph~$P_{\Q}$ for input QBFs of the form
  $Q=Q_1 V_1. {Q_2} V_2.$ $Q_3 V_3\cdots
    Q_{\ell+\log^*(\Card{\var(F)})}
    V_{\ell+\log^*(\Card{\var(F)})}. F,$ where is $F$ in 3-CNF
  (if~$Q_\ell=\exists$) or 3-DNF (if~$Q_\ell=\forall$).
\end{COR}
\begin{proof}
  Given a closed QBF of the form~$\Q'=Q_1 V'_1. {Q_2} V'_2. Q_3 V'_3$
  $\cdots Q_{\ell} V'_{\ell}. F'$, where~$\ell\geq 1$
  and~$F'$ is in 3-CNF if~$Q_\ell=\exists$,
  and~$F'$ is in 3-DNF if~$Q_\ell=\forall$,
  and~$k'$ is the $2$-local pathwidth
  of~$P_{\Q'}$.  Then, we apply our reduction~$R$ followed
  by~$R'$ on~$\Q'$ and iteratively apply~$R$
  and~$R'$.  We repeat this step
  exactly~$\log^*(k')$ many times and refer to the final result
  by~$\Q''$.  Note that the solutions to problem~\QBFSAT
  on~$\Q'$
  and~$\Q''$ are equivalent by Theorem~\ref{lem:correctness}.  Then,
  the resulting $2$-local pathwidth~$k''$ of~$P_{\Q''}$ is
  in~$\mathcal{O}(1)$ by Lemma~\ref{lem:compr}, i.e.,
  parameter~$k''$ is constant.  Hence,
  since~$\Q'$ is hard for~$\Sigma_\ell^P$,
  also~$\Q''$ is hard for~$\Sigma_\ell^P$, and
  $\Q''$ is~para-$\Sigma_{\ell}^P$-hard
  since~$k''$ is a constant.  Observe that~$k'\leq \Card{\var(F')}
  \leq \Card{\var(\matr(Q''))}$. %
  As a result, %
\QBFSAT for QBFs of the form~$\Q$ above is hard for~para-$\Sigma_\ell^P$.
\end{proof}
}

\vspace{-.5em}
\section{Conclusion} %
In this work, we presented a lower bound for deciding the validity {(\QBFSAT)} of quantified Boolean formulas (QBFs).
Thereby, we have significantly extended the current state-of-the-art of this line of research:
So far, \FIX{lower bound results under ETH} for \QBFSAT \FIX{parameterized by treewidth} were not available for all levels of the polynomial hierarchy.
The generalization of this result in Theorem~\ref{lab:primqbflb} {does} not only cover QBFs, parameterized by \FIX{treewidth and an arbitrary quantifier rank}, but solves a natural question for a well-known problem in complexity theory.
Interestingly, the result confirms the (asymptotic) optimality of the algorithm by Chen~\cite{Chen04a} for solving \QBFSAT and thereby answers a longstanding open question. 
Indeed, one cannot expect to solve \QBFSAT of quantifier rank~$\ell$ significantly better than in time~\FIX{$\Omega^*(\tower(\ell,k))$} in the treewidth~$k$.
The proof of this result relies on a novel reduction approach that makes use of a fragile balance between redundancy and structural dependency (captured by treewidth) and uses tree decompositions as a ``guide'' in order to achieve exponential \emph{compression} of the parameter treewidth. 
\FIX{
We encode \FIXCAM{dynamic programming on tree decompositions and obtain a technique for compressing treewidth.}
Note that both our technique and the results naturally carry over to path decompositions \FIX{and ($2$-local) pathwidth}.}

\FIX{Given the nature of our reduction, we observe that the reduction might also serve  in reducing treewidth in practice. 
In particular, solvers based on tree decompositions such as the QBF solver dynQBF~\cite{CharwatWoltran19} could benefit from significantly reduced treewidth; at the cost of increased quantifier rank by one. 
Since dynQBF is capable~\cite{LonsingEgly18} of solving instances up to treewidth~80 with quantifier rank more than two, slightly increasing the quantifier rank might be in practice a good trade-off for decreasing the treewidth significantly.}%

Another advantage of our reduction is that it gives rise to a versatile \emph{methodology} for showing lower bounds for arbitrary problems (depending on the ETH, parameterized by treewidth)
by reduction \FIX{from $\ell$-\QBFSAT, parameterized by treewidth as well}.
Thereby we avoid tedious reductions from \SAT (directly using ETH), which involves problem-tailored \FIX{gadgets to construct} instances whose treewidth is $\ell$-fold logarithmic in the number of variables or clauses of the given \SAT formula.
Further, we have listed a number of showcases to illustrate the applicability of this approach to natural problems that are beyond the second level of the polynomial hierarchy.
As a by-product we have established that \FIX{the canonical problems~$\text{\#}\Sigma_\ell\SAT$ and $\text{\#}\Pi_\ell\SAT$ of projected model counting} applied to \FIX{QBFs} {when parameterized by treewidth} always come at the price of \FIX{an additional level of exponentiality in the treewidth (compared to~$\ell$-\QBFSAT)}.
One direction for future work is to explore further problems parameterized by treewidth and to establish tightness of the so far existing upper bounds.
Another important direction is to work out techniques and showcases for \FIXCAM{``non-canonical''} lower bounds, where fptl-reductions are not sufficient and \FIX{using} a customized function $g$ is necessary.
Hence, our %
goal is to continue this line of research in order to use this toolkit for problems that do not \FIX{exhibit}~(e.g., \cite{LokshatanovMarxSaurabh2018}) \FIX{canonical runtimes, where fptl-reductions suffice}.
We hope this work will foster research and new insights 
on lower bounds.

\clearpage
\bibliography{qbf_tw_lbs-cleaned}

\clearpage
\appendix

\section{Appendix}
\BUGFIX{}{The following appendix contains additional proofs and proof details.
We plan on publishing this material along the submission (online).
\subsection{Detailed Proofs}
\begin{restatetheorem}[obs:uniqueremoval]
\begin{OBS}
Given an almost $c$-nice tree decomposition~$\mathcal{T}$ of a primal graph~$P_F$ for given 3-CNF or 3-DNF formula~$F$ of width~$k$. Then, one can easily create
a labeled almost $c$-nice tree decomposition~$\mathcal{T'}$ of~$P_F$ with a linear number of nodes in the number of nodes in~$\mathcal{T}$ such that~$\width(\mathcal{T}')=\width(\mathcal{T})$.
\end{OBS}
\end{restatetheorem}
\begin{proof}
The result follows since for each node one can have at most~$\binom{k}{3}\cdot 2^3$ many different clauses or terms of size at most~$3$, respectively.
\end{proof}

\subsection{Further Details on Remark~\ref{rem:qcsp}}

\subsection{Quantified CSP (QCSP)}
We define \emph{Constraint Satisfaction Problems (CSPs)} over finite domains 
and their evaluation similar to QBFs as follows. %
A CSP~$\mathcal{C}$ is a set of constraints~$C\in\mathcal{C}$
and we denote by~$\var(\mathcal{C})$ the set of \emph{(constraint) variables}
of~$\mathcal{C}$, which are over a (fixed) finite domain~$\mathcal{D}=\{0,\ldots\}$, which contains at least two values, i.e., $\mathcal{D}$ is at least ``Boolean''.
{%
  A \emph{constraint}~$C\in\mathcal{C}$ restricts certain variables~$\var(C)$
  and contains possible assignments to $\mathcal{D}$.
  Concretely, each~$c\in C$ is an \emph{allowed assignment}~$c: \var(C) \rightarrow \mathcal{D}$ 
  of each variable~$v\in\var(C)$
  to~$\mathcal{D}$. %
  Then, $\mathcal{C}$ is \emph{satisfiable}, if there is an assignment~$A: \var(\mathcal{C}) \rightarrow \mathcal{D}$ such that for every~$C\in \mathcal{C}$, there is an allowed assignment~$c\in C$
  such that assignment $A$, when restricted to domain~$\var(C)$, equals~$c$.
}%
  Let $\ell\geq 0$ be integer. A \emph{quantified CSP (QCSP)}~$\Q$ 
  is of the form
  $Q_{1} V_1.  Q_2 V_2.\cdots Q_\ell V_\ell. \mathcal{C}$ where
  $Q_i \in \{\forall, \exists\}$ for $1 \leq i \leq \ell$ and
  $Q_j \neq Q_{j+1}$ for $1 \leq j \leq \ell-1$; and where $V_i$ are
  disjoint, non-empty sets of constraint variables with
  $\bigcup^\ell_{i=1}V_i \subseteq \var(\mathcal{C})$; and $\mathcal{C}$ is a CSP.
  We call $\ell$ the \emph{quantifier rank} of~$Q$ and let
  $\matr(\Q)\eqdef \mathcal{C}$.
  Further, we denote the set~$\fvar(Q)$ of \emph{free variables}
  of~$Q$ by
  $\fvar(\Q)\eqdef \var(\matr(\Q)) \setminus (\bigcup^\ell_{i=1}V_i)$.
If~$\fvar(\Q)=\emptyset$, then~$\Q$ is referred to as~\emph{closed},
otherwise we say~$\Q$ is~\emph{open}. 
  The evaluation of QCSPs is defined as follows. %
  An \emph{assignment} is a mapping~$\iota: X \rightarrow \mathcal{D}$
  defined for a set~$X$ of variables.
  Given a QCSP~$\Q$ and an assignment~$\iota$, then~$\Q[\iota]$ is a
  QCSP that is obtained from~$\Q$, where every occurrence of
  any~$x\in \dom(\iota)$ in~$\matr(\Q)$ is replaced by~$\iota(x)$, and
  variables that do not occur in the result are removed from preceding
  quantifiers accordingly.
  A closed QCSP~$\Q$ \emph{evaluates to true}, or is valid, if~$\ell=0$
  and the CSP $\matr(Q)$ is satisfiable. Otherwise,
  i.e., if~$\ell \neq 0$, we distinguish according to~$Q_1$.
  If~$Q_1=\exists$, then~$\Q$ evaluates to true if and only if there
  exists an assignment~$\iota: V_1\rightarrow \mathcal{D}$ such
  that~$\Q[\iota]$ evaluates to true.  If~$Q_1=\forall$,
  then~$\Q[\iota]$ evaluates to true if for any
  assignment~$\iota: V_1 \rightarrow\mathcal{D}$, $\Q[\iota]$ evaluates to
  true.  An (open or closed) QCSP~$\Q$ is \emph{satisfiable} if
  there is an assignment~$\iota: \fvar(\Q) \rightarrow \mathcal{D}$
  such that resulting closed QCSP~$\Q[\iota]$ evaluates to
  true. Otherwise~$\Q$ is \emph{unsatisfiable}.
  For more details on QCSPs we refer to the literature~\cite{FergusonOSullivan07}. %

\subsubsection{Adapting reduction~$R$ for QCSP}
In the following, we list the constraints of reduction~$S$, 
which is similar to reduction~$R$, but
adapted to a QCSP instance~$\Q$, and consists of parts~$\mathcal{G}'$ (``guess''), 
$\mathcal{CK}'$ (``check'') as well as~$\mathcal{P}'$ (``propagate''). 
Note that for brevity, we denote constraints by a formula using equality $=$
between variables and elements of the domain~$\mathcal{D}$,
negation $\neg$, which inverts (in-)equality expressions, disjunction $\vee$ and conjunction $\wedge$,
which can be easily transformed into CSPs as defined above.
Further, we use for any variable~$v\in\var(\Q)$ (which is meant to be Boolean) the expression ``$v$'' , as a shortcut for~$v \neq 0$.

To this end, let~$\Q$ be a given QCSP of the form~$\Q\eqdef Q_1 V_1. {{Q_2}} V_2. \cdots \forall V_\ell. \mathcal{C}$, where each~$C\in\mathcal{C}$
uses exactly~$s$ many variables, i.e., 
we sketch the reduction for QCSPs $\Q$,
where $s$ is a constant such that $\Card{\var(C)} = s$.
Then, we assume a $c$-nice TD~$\mathcal{T} = (T, \chi, \delta)$, where~$T=(N, \cdot, \cdot)$ of the primal graph~$P_\Q$\footnote{The primal graph~$P_\Q$ for QCSP~$\Q$ is defined similar as for QBFs. Concretely, vertices of~$P_\Q$ 
are variables~$\var(\Q)$, and edges are between two variables,
whenever these variables occur in one common constraint~$C\in\matr(\Q)$.}.
In the following, we describe only the difference to the notation as defined for QBFs. 
In particular, we use notation as before, but on QCSPs,
e.g., for a set~$V\subseteq\var(\mathcal{C})$ of variables, we denote
by~$\lint{\mathcal{T}}{V} \eqdef \{x_t \mid x \in V, t \in
\nint{\mathcal{T}}{x}\}$.
For checking constraints, we need, similar to~$\lbvd{\mathcal{T}}$,
the sets~$\lbvs{\mathcal{T}} \eqdef \{b_{t,j}^0, \ldots, b_{t,j}^{\ceil{\log(\Card{\chi(t)} + 1)}-1} \mid t\in N, 1 \leq j \leq s\}$ that may again refer to a fresh element~$nil$,
and~$\lbvvs{\mathcal{T}} \eqdef \{v_{t,j} \mid t\in N, 1 \leq j \leq s\}$.
Again, we refer for each~$C\in\mathcal{C}$ 
to the first, second, \ldots and $s$-th variable of~$C$ by~$\ato{C}{1}, \ato{C}{2}$, \ldots, $\ato{C}{s}$, respectively.
Further, for given constraint~$C\in \mathcal{C}$, $c\in C$, and $1 \leq j \leq s$,
we let~$\bvs{c}{C}{j}$ denote $c(\ato{C}{j})$.

The adapted reduction~$S(\Q)$ creates~$\Q'\eqdef$
{\smallalign{\normalfont\small}
\begin{align*}
	&Q_1\ \lint{\mathcal{T}}{Q_1, V_1, \emptyset}.\ Q_2\ \lint{\mathcal{T}}{Q_2, V_2, V_1}.\ \cdots %
	\forall\ \lint{\mathcal{T}}{\forall, V_\ell, V_{\ell-1}}, \lbv{\mathcal{T}}.\notag\\
	&\exists\ \lint{\mathcal{T}}{\exists, \emptyset, V_{\ell}},\ \lbvv{\mathcal{T}}, \lbvvs{\mathcal{T}}, \lbvs{\mathcal{T}}.\ \mathcal{C}',%
\end{align*}}%
where~$\mathcal{C}'$ is a CSP and consists of sets~$\mathcal{G}'$, $\mathcal{CK}'$, and~$\mathcal{P}'$ as follows.

\medskip
\noindent\textbf{Guess part~$\mathcal{G}'$.}
{
\smallalign{\normalfont\small}
\begin{align}
	\label{redq:guessatom}&\bigvee_{b \in \bval{x}{t}} \neg b \vee x_t{=}v_t &\pushright{\text{for each } x_t\in \lint{\mathcal{T}}{\var(\mathcal{C})}}\raisetag{1.15em}\\
	\label{redq:guessatomj}&\bigvee_{b \in \bvali{x}{t}{j}} \neg b \vee x_t{=}v_{t,j}&\pushright{\text{for each } x_t\in \lint{\mathcal{T}}{\var(\mathcal{C})}, 1 \leq j \leq s}\raisetag{1.15em}%
\end{align}
}

\medskip
\noindent\textbf{Check part~$\mathcal{CK}'$.}

{
\smallalign{\normalfont\small}
\begin{align}
	\label{q:chk_satcb}&\pushleft{\bigvee_{c\in C}[\bigwedge_{1\leq j \leq s, b\in\bvali{x}{t}{j}} b \wedge v_{t,j}{=}\bvs{c}{C}{j}]}&\pushleft{\text{for each } C\in \mathcal{C}, t=\lnode{C}}%
\end{align}
}%

\medskip
\noindent\textbf{Propagate part~$\mathcal{P}'$.}
{
\smallalign{\normalfont\small}
\begin{align}
	\label{redq:propatom}&\bigvee_{b \in \bval{x}{t}} \neg b \vee \bigvee_{b' \in \bval{x}{t_i}} \neg b' \vee v_t{=}v_{t_i}&\pushleft{\text{for each } t\in N, t_i\in\children(t),} \notag\vspace{-1.5em}\\ &&\pushleft{x\in \chi(t)\cap\chi(t_i)}\raisetag{1.15em}%
\end{align}
}%

{
\smallalign{\normalfont\small}
\begin{align}
	\label{redq:propatomj}&\bigvee_{b \in \bval{x}{t}} \neg b \vee \bigvee_{b' \in \bvali{x}{t}{j}} \neg b' \vee v_t{=}v_{t,j} &\pushleft{\text{for each } C\in \mathcal{C}, 1\leq j\leq s, }\notag\vspace{-1.65em}\\
	&&\pushleft{t=\lnode{C}, x=\ato{C}{j}}\notag\vspace{-.2em}\\
\end{align}
}%

Using this reduction~$S$, one obtains similar results as for QBFs.
In particular, under ETH, validity of QCSPs~$\Q$ over domain~$\mathcal{D}$ of 
quantifier rank~$\ell$, where~$k=\pw{P_\Q}$,
cannot be decided in time $\tower(\ell - 1,|\mathcal{D}|^{o(k)})\cdot \BUGFIX{\poly(\Card{\var(\Q)})}{|\mathcal{D}|^{o(\Card{\var(\Q)})}}$.
}

\subsection{Counting Problems}
A \emph{witness function} is a
function~$\mathcal{W}\colon \Sigma^* \rightarrow 2^{{\Sigma'}^*}$ that
maps an instance~$I \in \Sigma^*$ to a finite subset
of~${\Sigma'}^*$. We call the set~$\WWW(I)$ the \emph{witnesses}. A
\emph{parameterized counting
  problem}~$L: \Sigma^* \times \Nat \rightarrow \Nat_0$ is a
function that maps a given instance~$I \in \Sigma^*$ and an
integer~$k \in \Nat$ to the cardinality of its
witnesses~$\Card{\WWW(I)}$.
Let $\mtext{C}$ be a decision complexity class,~e.g., $\Ptime$. Then,
$\cntc\mtext{C}$ denotes the class of all counting problems whose
witness function~$\WWW$ satisfies (i)~there is a
function~$f: \Nat_0 \rightarrow \Nat_0$ such that for every
instance~$I \in \Sigma^*$ and every $W \in \WWW(I)$ we have
$\Card{W} \leq f(\CCard{I})$ and $f$ is computable in
time~$\bigO(\poly(\CCard{I}))$ and (ii)~for every
instance~$I \in \Sigma^*$ the decision problem~$\WWW(I)$ belongs to
the complexity class~$\mtext{C}$.
Then, $\cntc\Ptime$ is the complexity class consisting of all counting
problems associated with decision problems in \NP.
Let $L$ and $L'$ be counting problems with witness functions~$\WWW$
and $\WWW'$. A \emph{parsimonious reduction} from~$L$ to $L'$ is a
polynomial-time reduction~$r: \Sigma^* \rightarrow \Sigma'^*$ such
that for all~$I \in \Sigma^*$, we
have~$\Card{\WWW(I)}=\Card{\WWW'(r(I))}$. It is easy to see that the
counting complexity classes~$\cntc\mtext{C}$ defined above are closed
under parsimonious reductions. It is clear for counting problems~$L$
and $L'$ that if $L \in \cntc\mtext{C}$ and there is a parsimonious
reduction from~$L'$ to $L$, then $L' \in \cntc\mtext{C}$.

\subsection{Problem Compendium}\label{sec:compendium}

\paragraph*{Projected Model Counting (\PQBF).}
\citex{DurandHermannKolaitis05} have shown that the canonical problems $\text{\#}\Sigma_\ell\SAT$, and $\text{\#}\Pi_\ell\SAT$
are $\cntc\Sigma^\Ptime_\ell$-complete, and~$\cntc\Pi^\Ptime_\ell$-complete (both via parsimonious reductions), respectively.
Note that both problems are special cases of~\PQBF as defined in Section~\ref{sec:prelim}.
Further, \sharpSAT is the problem~$\text{\#}\Sigma_0\SAT$, and \PMC is the problem~$\text{\#}\Sigma_1\SAT$.
\dproblem{$\text{\#}\Sigma_\ell\SAT$~\cite{DurandHermannKolaitis05}}{Open QBF~$Q=\exists V_1. \forall V_2. \cdots Q_\ell V_\ell F$, where $F$ is a
  Boolean formula in 3-DNF if $Q_\ell=\forall$ (and 3-CNF if $Q_\ell=\exists$).}{Number of
  assignments~$\iota: \fvar(\Q) \rightarrow \{0,1\}$ such that~$Q[\iota]$ is valid.}
\dproblem{$\text{\#}\Pi_\ell\SAT$~\cite{DurandHermannKolaitis05}}{Open QBF~$Q=\forall V_1. \exists V_2. \cdots Q_\ell V_\ell F$, where $F$ is a
  Boolean formula in 3-DNF if $Q_\ell=\forall$ (and 3-CNF if $Q_\ell=\exists$).}{Number of
  assignments~$\iota: \fvar(\Q) \rightarrow \{0,1\}$ such that~$Q[\iota]$ is valid.}
\paragraph*{Abstract Argumentation.}

We consider the Argumentation Framework by \citex{Dung95a}.
An \emph{argumentation framework~(AF)} is a directed graph~$F=(A, R)$
where $A\neq\emptyset$ is a finite set of arguments and
$R \subseteq A\times A$ a pair of arguments representing direct attacks of arguments.
In argumentation, we interest in computing so-called
\emph{extensions}, which are subsets~$S \subseteq A$ of the arguments
that meet certain properties according to certain semantics as given
below.
An argument~$s \in S$, is called \emph{defended by $S$ in $F$} if for
every $(s', s) \in R$, there exists $s'' \in S$ such that
$(s'', s') \in R$.  The family~$\adef_F(S)$ is defined by
$\adef_F(S) \eqdef\SB s \SM s \in A, s \text{ is defended by $S$ in
  $F$} \SE$.
We say $S \subseteq A$ is \emph{conflict-free in~$S$} if
$(S\times S) \cap R = \emptyset$; $S$ is \emph{admissible ($\admissible$) in $F$} if
(i) $S$ is \emph{conflict-free in $F$}, and
(ii) every $s \in S$ is \emph{defended by $S$ in $F$}.
Assume an admissible set~$S$.
Then,
(iiia) $S$ is~\emph{preferred ($\preferred$) in~$F$}, if there is no $S' \supset S$ that is \emph{admissible in $F$};
(iiib) $S$ is \emph{semi-stable ($\semistable$) in $F$} if there is no admissible set $S' \subseteq A$ in~$F$ with~$S^+_R\subsetneq (S')^+_R$ where $S^+_R:=S\cup\SB a\SM (b,a)\in R, b \in S \SE$;
a conflict-free set~$S$ is \emph{stage in $F$} if there is no conflict-free set~$S'\subseteq A$ in~$F$ with~$S^+_R\subsetneq (S')^+_R$.
Let $\ALL$ abbreviate the set $\{\preferred, \semistable, \stage\}$.
For a semantics~$\SEM \in \ALL$, $\SEM(F)$ denotes the set of
\emph{all extensions} of semantics~\SEM in $F$.
Let $\SEM \in \ALL$ be an abstract argumentation
semantic, $F=(A,R)$ be an argumentation framework, and $a \in A$ an
argument. 
The \emph{credulous reasoning problem} $\Cred_\SEM$ asks to decide
whether there is an \SEM-extension of $F$ that contains the
(credulous) argument~$a$. 
The \emph{skeptical reasoning problem} $\Skep_\SEM$ asks
to decide whether all \SEM-extensions of $F$ contain the argument~$a$.
\dproblem{Projected Credulous Counting~$\PCC_\SEM$~\cite{FichteHecherMeier18c}}{$\SEM\in\ALL$, argumentation framework~$F$, set~$P$ of projection arguments, argument~$a\in A$.}{Number of $\SEM$-extensions,
  where~$a$ is credulously accepted, restricted to~$P$, in more detail,
  $\Card{ \SB S \cap P \SM S \in \SEM(F), a \in S \SE}$.}  

\paragraph*{Propositional Abduction.}

The \emph{Propositional Abduction Problem (\PAP)} consists of a tuple~$(T,H,M)$, where~$T$ is a Boolean formula (\emph{theory}) in CNF, $H\subseteq\var(T)$ is a set of \emph{hypothesis}, and~$M\subseteq\var(T)$ forms a set of \emph{manifestations}.
A set~$S\subseteq H$ is a \emph{solution} to the problem~$(T,H,M)$ if there is an assignment~$\iota: \var(T) \rightarrow \{0,1\}$ such that $(T\cup S)[\iota]$ evaluates to true, and for every assignment~$\iota: \var(T) \rightarrow \{0,1\}$, where~$(T\cup S)[\iota]$ evaluates to true, also $M[\iota]$ evaluates to true.
Given a \PAP~$(T,H,M)$. Then, let~$\mathcal{S}(T,H,M)$ denote the set of solutions to this problem.

\dproblem{Projected Propositional Abduction Counting~$\PPAP$}{\PAP problem~$(T,H,M)$, set~$P\subseteq\var(T)$ of projection variables.}{Number of solutions to~$(T,H,M)$ where only the number of combinations with respect to~$P$ are of interest, i.e.,
  $\Card{ \SB S \cap P \SM S \in \mathcal{S}(T,H,M) \SE}$.} 

\paragraph*{Answer Set Programming (ASP).}

We follow standard definitions of \FIX{Boolean} disjunctive ASP. For
comprehensive foundations, we refer to introductory
literature~\cite{BrewkaEiterTruszczynski11,JanhunenNiemela16a}.
Let $\ell$, $m$, $n$ be non-negative integers such that
$\ell \leq m \leq n$, $a_1$, $\ldots$, $a_n$ be Boolean variables.
A \emph{program}~$\prog$ is a set of \emph{rules} of the form
\(
a_1\por \cdots \por a_\ell \hsep a_{\ell+1}, \ldots, a_{m}, \neg
a_{m+1}, \ldots, \neg a_n.
\)
For a rule~$r$, we let $H_r \eqdef \{a_1, \ldots, a_\ell\}$,
$B^+_r \eqdef \{a_{\ell+1}, \ldots, a_{m}\}$, 
$B^-_r \eqdef \{a_{m+1}, \ldots, a_n\}$, and~$\var(r)\eqdef H_r\cup B^+_r \cup B^-_r$. Consequently, $\var(\prog)\eqdef \bigcup_{r\in\prog}\var(r)$.

An \emph{interpretation} $I\subseteq\var(\prog)$ is a set of variables. $I$ \emph{satisfies} a
rule~$r$ if $(H_r\,\cup\, B^-_r) \,\cap\, I \neq \emptyset$ or
$B^+_r \setminus I \neq \emptyset$.  $I$ is a \emph{model} of $\prog$
if it satisfies all rules of~$\prog$. %
The \emph{Gelfond-Lifschitz
  (GL) reduct} of~$\prog$ under~$I$ is the program~$\prog^I$ obtained
from $\prog$ by first removing all rules~$r$ with
$B^-_r\cap I\neq \emptyset$ and then removing all~$\neg z$ where
$z \in B^-_r$ from the remaining
rules~$r$~\cite{GelfondLifschitz91}. %
$I$ is an \emph{answer set} of a program~$\prog$ if $I$ is a subset-minimal
model of~$\prog^I$. %

\eproblem{Answer Set Consistency (\cASP)}{Program~$\prog$.}{Does there exist an answer set of program~$\prog$?}
\dproblem{Answer Set Counting (\sharpASP)}{Program~$\prog$.}{Number of answer sets of~$\prog$.}
\dproblem{Projected Answer Set Counting
  (\PASP)~\cite{Aziz15a}}{Program~$\prog$, set~$P\subseteq\var(\prog)$
  of projection variables.}{Number of answer sets projected to the
  set~$P$, which is,
  $\Card{\{S\cap P \mid S \text{ is an answer set of }\prog\}}$.}

\paragraph*{Epistemic Logic Programming (ELP).}

Let $\ell$, $m$, $n$ be non-negative integers such that
$\ell \leq m \leq n$, $a_1$, $\ldots$, $a_n$ be Boolean variables.
Then, an \emph{epistemic program}~$\prog$ is a set of \emph{epistemic rules} of the form
\(
a_1\por \cdots \por a_\ell \hsep l_{\ell+1}, \ldots, l_{m}, \neg
l_{m+1}, \ldots, \neg l_n.
\),
where~$l_i$ for~$\ell + 1 \leq i \leq n$
is either variable~$a_i$, 
or an \emph{epistemic literal} potentially using \emph{epistemic negation~$\mathsf{not}$} 
of the form~$\mathsf{not}\ a_i$ or $\mathsf{not}\ \neg a_i$.
Let~$\ep(\prog)$ be the set of epistemic literals occurring in~$\prog$.
Given a \emph{guess}~$\sigma\subseteq\ep(\prog)$.
The \emph{epistemic reduct}~$\prog^\sigma$ corresponds to the program~$\prog$,
where each occurrence of~$\mathsf{not}\ l\in\sigma$ is replaced by truth constant~$\top$,
and each occurrence of~$\mathsf{not}\ l\in\ep(\prog)\setminus\sigma$ is replaced by~$\neg l$.
Note that~$\prog^\sigma$ corresponds to an ASP program, assuming that~$\top$ is
in any answer set and potential double negation ($\neg\neg$) is treated~\cite{FaberPfeiferLeone11} 
accordingly\footnote{Note that there are several ways to treat potential occurrences of double negation, which are out of the scope of this work.}.
Then, given guess~$\sigma\subseteq \ep(\prog)$,
the collection~$\mathcal{A}$ of \emph{all} the answer sets of~$\prog^\sigma$ is a \emph{candidate world view} of~$\prog$ w.r.t.\ $\sigma$,
if (i)~$\mathcal{A}\neq\emptyset$, (ii) for each~$\mathsf{not}\ l\in\sigma$ there is an answer set~$A\in\mathcal{A}$ 
where~$A\cup\{l\}$ is unsatisfiable ($A$ is viewed as a formula of facts), and (iii) for each~$\mathsf{not}\ l\in\ep(\prog)\setminus\sigma$ and each answer set~$A\in\mathcal{A}$,
$A\cup\{l\}$ is satisfiable.
Further, set~$\mathcal{A}$ of all the answer sets of~$\prog$ is referred to by \emph{world view} of~$\prog$ w.r.t.\ $\sigma$,
if there is no set~$\sigma'$ that is a candidate world view of~$\prog$ w.r.t.\ $\sigma'$ such that~$\sigma \subsetneq \sigma' \subseteq \ep(\prog)$.

\eproblem{Candidate World View Check~\cite{EiterShen17}}{ Epistemic
  program~$\prog$, variable~$a\in\var(\prog)$.}{
  Does there exist a guess~$S\subseteq \ep(\prog)$ that induces a candidate world view, more precisely, %
  $\{S \mid S \subseteq \ep(\prog), \text{ set }\mathcal{A}\text{ of
    answer } \text{sets of }\prog^S$ $\text{ is a candidate world view}\}\neq\emptyset$?}

\eproblem{World View Check~\cite{EiterShen17}}{
Epistemic program~$\prog$, variable~$a\in\var(\prog)$.}{
Does there exist a guess~$S\subseteq \ep(\prog)$ that induces a world view, 
where~$a$ is true, i.e., 
$\{S \mid S \subseteq \ep(\prog), \text{ set }\mathcal{A}\text{ of answer } \text{sets of }\prog^S$ $\text{ is a world view, }$ $a \in A \text{ for }\text{any }$ $A\in\mathcal{A}\}\neq\emptyset$?}

 \dproblem{\#Projected Guesses to World Views}{
Epistemic program~$\prog$, variable~$a\in\var(\prog)$, set~$P\subseteq\var(\prog)$ of projection variables.}{
Number of distinct epistemic literal sets~$\ep(\prog)$ projected to~$P$ that lead to world views, where~$a$ is true, i.e., 
$|\{S\cap \{\mathsf{not}\ p, \mathsf{not}\ \neg p \mid p \in P\} \mid S \in \ep(\prog), \text{ the set }\mathcal{A}\text{ of answer sets of }$ $\prog^S \text{ is a world view, }a\in A$ $\text{for any }$ $A\in\mathcal{A}\}|$.}

\end{document}